\documentclass[draftcls, onecolumn, 11pt]{IEEEtran}
\usepackage{graphics, subfigure, theorem, times, amsfonts, amsmath, amssymb, cite, mathtools}
\usepackage{tikz, epic,eepic}
\usetikzlibrary{shapes,arrows}
\usepackage{pgfplots}
\usepackage{color}
\usepackage{hyperref,url}
\definecolor{darkblue}{rgb}{0,0.08,0.4} 
\definecolor{darkred}{rgb}{0.4,0.08,0} 
\input{mysymbol.sty}

\newtheorem{theorem}{Theorem}

\newtheorem{lemma}{Lemma}

{\itshape}{\rmfamily}

\newtheorem{remark}{\hspace{0pt}\bf Remark}

\usepackage{hyperref,url,setspace}
\usepackage{algorithm,algorithmic}
\newcommand{\lbra}{$[$}
\newcommand{\rbra}{$]$}

\newcommand{\BR}{\text{BR}}

\def\forall{\text{for all\ }}

\title{Bayesian Quadratic Network Game Filters}
\author{\small \IEEEauthorblockN{Ceyhun Eksin, Pooya Molavi, Alejandro Ribeiro, and Ali Jadbabaie} \\ 
Department of Electrical and Systems Engineering, University of Pennsylvania, \\ 200 South 33rd Street, Philadelphia, PA 19104. \\ Email:\{ceksin, pooya, aribeiro, jadbabai\}@seas.upenn.edu.\thanks{Work in this paper supported by ARO W911NF-10-1-0388, NSF CAREER CCF-0952867, NSF CCF-1017454, and AFOSR MURI FA9550-10-1-0567. Part of the results in this paper have been submitted to ICASSP 2013 \cite{cEksinEtal13}.}}

\date{\today}

\begin{document}

\normalsize
\maketitle
%\doublespace
\begin{abstract}
A repeated network game where agents have quadratic utilities that depend on information externalities -- an unknown underlying state -- as well as payoff externalities -- the actions of all other agents in the network -- is considered. Agents play Bayesian Nash Equilibrium strategies with respect to their beliefs on the state of the world and the actions of all other nodes in the network. These beliefs are refined over subsequent stages based on the observed actions of neighboring peers. This paper introduces the Quadratic Network Game (QNG) filter that agents can run locally to update their beliefs, select corresponding optimal actions, and eventually learn a sufficient statistic of the network's state. The QNG filter is demonstrated on a Cournot market competition game and a coordination game to implement navigation of an autonomous team.\vspace{-5mm}
\end{abstract}

%!TEX root = bqng_arxiv.tex
\section{Introduction}

Games with information and payoff externalities are common models of networked economic behavior. In, e.g., trade decisions in a stock market, the payoff that a player receives depends not only on the fundamental (unknown) price of the stock but on the buy decisions of other market participants. Thus, players must respond to both, their belief on the price of the stock and their belief on the actions of other players \cite{MorrisShin}. Similar games can also be used to model the coordination of members of an autonomous team whereby agents want to select an action that is jointly optimal but only have partial knowledge about what the action of other members of the team will be. Consequently, agents select actions that they deem optimal given what they know about the task they want to accomplish and the actions they expect other agents to take. 

In both of the examples in the previous paragraph we have a network of autonomous agents intent on selecting actions that maximize local utilities that depend on an unknown state of the world -- information externalities --  and  also the unknown actions of all other agents -- payoff externalities. In a Bayesian setting -- or a rational setting, to use the nomenclature common in the economics literature \cite{sobel2000economists} -- nodes form a belief on the actions of their peers and select an action that maximizes the expected payoff with respect to those beliefs. In turn, forming these beliefs requires that each network element make a model of how other members will respond to their local beliefs. The natural assumption is that they exhibit the same behavior, namely that they are also maximizing their expected payoffs with respect to a model of other nodes' responses. But that means the first network element needs a model of other agents' models which shall include their models of his model of their model and so on. The fixed point of this iterative chain of reasoning is a Bayesian Nash Equilibrium (BNE). 

In this paper we consider repeated versions of this game in which agents observe the actions taken by neighboring agents at a given time. In observing neighboring actions agents have the opportunity to learn about the private information that neighbors are, perhaps unwillingly, revealing \cite{MorrisShin}. Acquiring this information alters agents' beliefs leading to the selection of new actions which become known at the next play prompting further reevaluation of beliefs and corresponding actions. In this context we talk of Bayesian learning because the agents' goal can be reinterpreted as the eventual learning of peers' actions so that expected payoffs coincide with actual payoffs. This paper considers Gaussian prior distributions and quadratic utilities. For this type of problem we introduce the Quadratic Network Game (QNG) filter that agents can run locally to update their beliefs, select corresponding actions that maximize expected payoffs, and eventually learn a sufficient statistic of the network's state. 

The burden of computing a BNE in repeated games is, in general, overwhelming even for small sized networks \cite{mossel2010efficient}. This intractability has led to the study of simplified models in which agents are non-Bayesian and update their beliefs according to some heuristic rule \cite{BalaGoyal,Alireza,Zwiebeletal,GolubJackson,DeGroot}. A different simplification is obtained in models with pure information externalities where payoffs depend on the self action and an underlying state but not on the actions of others. This is reminiscent of distributed estimation \cite{Xiaoetal,Rabbat,Ribeiro_ConsensusI,SoummyaKar_e,aSayed,chenSayed,ChenSayed_conf,RadTahbaz-Salehi,krishnamurthyquickest,Ribeiro10c} since agents deduce the state of the world by observing neighboring actions without strategic considerations on the actions of peers. Computations are still intractable in the case of pure information externalities and for the most part only asymptotic analyses of learning dynamics with rational agents are possible \cite{GaleKariv,Rosenbergetal,Djuric}. Explicit methods to maximize expected payoffs given all past observations of neighboring actions are available only when signals are Gaussian \cite{mossel2010efficient} or when the network structure is a tree \cite{Yash}. For the network games considered here in which there are information as well as payoff externalities, not much is known besides asymptotic analyses of learning dynamics \cite{linear_games_Allerton, asilomar_paper,EksinEtal13_b}.

The specific setting considered in this paper is introduced in Section \ref{the_model}. Agents repeatedly play a game whose payoffs are represented by a utility function that is quadratic in the actions of all agents and an unknown real-valued parameter. At the start of the game each agent makes a private observation of the unknown parameter corrupted by additive Gaussian noise.  At each play stage agents observe actions of adjacent peers from the previous stage that they incorporate into a local observation history which they use to update their inference of the unknown parameter, and synchronously take actions that maximize their expected payoffs. Actions that maximize expected payoffs with respect to local observations histories are defined as best responses to the expected actions taken by other agents. When the expected actions of other agents are also modeled as best responses with respect to their respective observation histories, we say that the network settles into a BNE (Section \ref{sec_bne_definition}). 

In Section \ref{recursion} we determine a mechanism to calculate BNE actions from the perspective of an outside clairvoyant observer that knows all private observations. For this clairvoyant observer the trajectory of the game is completely determined but individual agents operate by forming a belief on the private signals of other agents. We start from the assumption that this probability distribution is normal with an expectation that, from the perspective of the outside observer, can be written as a linear combination of the actual private signals. If such is the case, we prove that there exists a set of linear equations that can be solved to obtain actions that are linear combinations of estimates of private signals (Lemma \ref{linear_equilibrium_at_time_t}). This is then used to show that after observing the actions of their respective adjacent peers the probability distributions on private signals of all agents remain Gaussian with expectations that are still linear combinations of the actual private signals (Lemma \ref{rational_belief_updates_theorem_L}). We proceed to close a complete induction loop to derive a recursive expression that the outside clairvoyant observer can use to compute BNE actions for all game stages (Theorem \ref{rational_updates_theorem}). 

In Section \ref{QNG_filter} we leverage the recursion derived in Section \ref{recursion} to derive the QNG filter that agents can run locally, i.e., without access to all private signals, to compute their BNE action. Results in sections \ref{recursion} and \ref{QNG_filter} are generalized to the case of vector states and observations (Section \ref{vector_state}). We apply the scalar QNG filter to a Cournot competition model (Section \ref{cournot}) and to the coordinated movement of a team of mobile agents (Section \ref{coordination_game}).

\medskip\noindent{\bf Notation.} Vectors $\bbv\in\reals^n$ are written in boldface and matrices $A\in\reals^{n\times m}$ in uppercase. We use $\bbzero$ to denote all-zero matrices or vectors of proper dimension. If the dimension is not clear from context, we specify $\bbzero_{n\times m}$. We use $\bbone$ to denote all-one matrices or vectors of proper dimension and $\bbone_{n\times m}$ to clarify dimensions. We use $\bbe_i$ to denote the $i$th element of the standard orthonormal basis of $\reals^n$ and $\bar\bbe_i:=\bbone-\bbe_i$ to write an all-one vector with the $i$th component nulled.

%!TEX root = bqng_arxiv.tex
\section{Gaussian Quadratic Games} \label{the_model}

We consider games with incomplete information in which $N$ identical agents in a network repeatedly choose actions and receive payoffs that depend on their own actions, an unknown scalar parameter $\theta\in\reals$, and actions of all other agents. The network is represented by an undirected connected graph $G = (V,E)$ with node set $V=1,\ldots, N$ and edge set $E$. The network structure restricts the information available to agent $i$ who is assumed to observe actions of agents $j$ in its neighborhood $n(i) := \{j: \{j,i\} \in E\}$ composed of agents that share an edge with him. The degree of node $i$ is given by the cardinality of the set $n(i)$ and denoted as $d(i):=\#n(i)$. The neighbors of $i$ are denoted $j_{i,1}<,\ldots,<j_{i,d(i)}$. We assume the network graph $G$ is known to all agents. 

At time $t = 0$ agent $i$ observes a private signal $x_i \in\reals$ which we model as being given by the unknown parameter $\theta$ contaminated with zero mean additive Gaussian noise $\eps_i$,
\begin{equation}\label{private_signal_model}
x_i = \theta + \epsilon_i.
\end{equation}
The noise variances are denoted as $c_{i}:=\bbE\left[\eps_1^2\right]$ and grouped in the vector $\bbc:=[c_1,\ldots,c_N]^T$ which is assumed known to all agents.  The noise terms $\eps_i$ are further assumed independent across agents. For future reference define the vector of private signals $\bbx := [x_1, \ldots, x_N]^T \in \reals^{N \times 1}$ grouping all local observations.

Consider a discrete time variable $t=0,1,2,\ldots$ to index subsequent stages of the game. At each stage $t$ agent $i$ takes scalar action $a_i(t)\in\reals$. The selection of agent $i$, along with the concurrent selections $a_j(t)$ of all other agents $j\in V\setminus\{i\}$ results in a 
payoff $u_i(a_i(t), \{a_{j}(t)\}_{j \in V \setminus i}, \theta)$ that agent $i$ wants to make as large as possible. In this paper we restrict attention to quadratic payoffs which for simplicity we assume to be time invariant. Specifically, selection of actions $\{a_i = a_i(t)\}_{i\in V}$ when the state of the world is $\theta$ results in agent $i$ experiencing a reward
\begin{align} \label{quadratic_utility_form}
    u_i(a_i, \{a_{j}\}_{j \in V \setminus i}, \theta)  
         := - \frac{1}{2}  a_i^2 
            + \sum_{j \in V\setminus i} \beta_{ij} a_i a_j 
            + \delta a_i \theta ,
\end{align}
where $\beta_{ij}\in\reals$ for all $i \in V$, $j \in V \setminus i$ and $\delta\in\reals$ are real valued constants. Notice that since $\partial^2 u_i/\partial a_i^2 = -1 < 0$, the payoff function in \eqref{quadratic_utility_form} is strictly concave with respect to the self action $a_i$ of agent $i$.

Although the goal of agent $i$ is to select the action $a_i(t)$ that maximizes the payoff in \eqref{quadratic_utility_form}, this is not possible because neither the state $\theta$ nor the actions $\{a_{j}(t)\}_{j \in V \setminus i}$ are known to him. Rather, agent $i$ needs to reason about state $\theta$ and actions $\{a_{j}(t)\}_{j \in V \setminus i}$ based on its available information. At time $t=0$ only the private signal $x_i$ is known. Define then the initial information as $h_{i,0} = \{x_i\}$. The information $h_{i,0}$ is used to reason about $\theta$ and the initial actions $\{a_{j}(0)\}_{j \in V \setminus i}$ that other agents are to take in the initial stage of the game. At the playing of this stage, agent $i$ observes the actions $\bba_{n(i)}(0):=[a_{j_{i,1}}(0),\ldots,a_{j_{i,d(i)}}(0)]^T \in \reals^{d(i) \times 1}$ of all agents in his neighborhood. These observed neighboring actions become part of the observation history $h_{i,1} = \big\{ x_i, \bba_{n(i)}(0)\} = \big\{ h_{0,i}, \bba_{n(i)}(0)\big\}$ which allows agent $i$ to improve on his estimate of $\theta$ and the actions $\{a_{j}(1)\}_{j \in V \setminus i}$ that other agents will play on the first stage of the game, thereby also affecting the selection of its own action $a_i(1)$. In general, at any point in time $t$ the history of observations $h_{i,t}$ is augmented to incorporate the actions of neighbors in the previous stage,
\begin{equation}\label{eqn_history_definition}
    h_{i,t} := \big\{ h_{i,t-1}, \bba_{n(i)}(t-1)\big\} 
             = \big\{ x_i, \bba_{n(i)}(u), u<t \big\}. 
\end{equation}
The observed action history $h_{i,t}$ is then used to update the estimates of the world state $\theta$ and the upcoming actions $\{a_j(t)\}_{j \in V \setminus i}$ of all other agents leading to the selection of the action $a_{i}(t)$ in the current stage of the game.
 
The final components of the game that we introduce are the strategies $\sigma_{i,t}$ that are used to map histories to actions. %In general, the strategy $\sigma_{i,t}$ of agent $i$ at time $t$ is a random variable measurable with respect to the history $h_{i,t}$. 
In this paper we focus on pure strategies that can be written as functions that map history realizations $h_{i,t}$ to actions $a_{i}(t)$ 
\begin{equation}\label{eqn_strategy_definition}   
\sigma_{i,t}: h_{i,t}\mapsto a_{i}(t). 
\end{equation}
We emphasize the difference between strategy and action. An action $a_i(t)$ is the play of agent $i$ at time $t$, whereas strategies $\sigma_{i,t}$ refer to the map of histories to actions. We can think of the action $a_{i}(t) = \sigma_{i,t} (h_{i,t})$ as the value of the strategy function $\sigma_{i,t}$ associated with the given observed history $h_{i,t}$. Further define the strategy of agent $i$ as the concatenation $\sigma_i:=\{\sigma_{i,u}\}_{u=0,\dots, \infty}$ of strategies that agent $i$ plays at all times. Use $\sigma_t:=\{\sigma_{i,t}\}_{i \in V}$ to refer to the strategies of all players at time $t$, $\sigma_{0:t}:=\{\sigma_{u}\}_{u=0,\ldots t}$ to represent the strategies played by all players between times $0$ and $t$, and $\sigma:=\{\sigma_{u}\}_{u=0,\dots, \infty}=\{\sigma_i\}_{i \in V}$ to denote the strategy profile for all agents $i \in V$ and times $t$. As in the case of the network topology, the strategy $\sigma$ is also assumed to be known to all agents. We study mechanisms for the construction of strategies in the following section.

\subsection{Bayesian Nash equilibria}\label{sec_bne_definition}

Given that agent $i$ wants to maximize the utility in \eqref{quadratic_utility_form} but has access to the partial information available in the observed history $h_{i,t}$ in \eqref{eqn_history_definition}, a reasonable strategy $\sigma_{i,t}$ is to select the action  $a_i(t)$ that maximizes the expected utility with respect to the history $h_{i,t}$. To write this formally note that this expected utility depends on strategies $\sigma_{0:t-1}$ played in the past by all agents and on strategies $\{\sigma_{j,t}\}_{j \in V \setminus i}$ that all other agents are to play in the upcoming turn. Fix then the past strategies $\sigma_{0:t-1}$ and the upcoming strategies $\{\sigma_{j,t}\}_{j \in V \setminus i}$ of other players and define the corresponding best response of player $i$ at time $t$ as
\begin{align}\label{eqn_best_response}
    &\BR_{i,t}\Big( \sigma_{0:t-1},\ \{\sigma_{j,t}\}_{j \in V \setminus i}\Big) 
%                         \\ \nonumber & \hspace{10mm}
        := \argmax_{a_i\in \reals}\
             \bbE_{\sigma_{0:t-1}}\Big[u_i(a_i,\{\sigma_{j,t}(h_{j,t})\}_{j\in V\setminus i},\theta)
                  \given h_{i,t}\Big]  .
\end{align}
The strategies $\sigma_{0:t-1}$ in \eqref{eqn_best_response} played at previous times mapped respective histories $\{h_{j,u}\}_{j\in V}$ to actions $\{a_j(u)\}_{j \in V}$ for $u<t$. Therefore, the past strategies $\sigma_{0:t-1}$ determine the manner in which agent $i$ updates his beliefs on the state of the world $\theta$ and on the histories $\{h_{j,t}\}_{j\in V\setminus i}$ observed by other agents. As per \eqref{eqn_strategy_definition} the strategy profiles $\{\sigma_j(t)\}_{j\in V\setminus i}$ of other players in the current stage permit transformation of history beliefs $\{h_{j,t}\}_{j\in V\setminus i}$ into a probability distribution over respective upcoming actions $\{a_j(t)\}_{j \in V \setminus i}$. The resulting joint distribution on $\{a_j(t)\}_{j \in V \setminus i}$ and $\theta$ permits evaluation and maximization of the expectation in \eqref{eqn_best_response}. 

One can think of the profiles $\{\sigma_j(t)\}_{j\in V\setminus i}$ played by other agents in the upcoming stage as the model agent $i$ makes of the behavior of other agents. In that sense the sensible assumption is that other agents are also playing best response to a best response model of other agents. I.e., agent $i$ assumes agent $j$ is playing the best response to its respective model of the behavior of other agents and that the model agent $j$ makes of these responses is that these agents also play best response to a best response model. This modeling assumption leads to the definition of Bayesian Nash equilibria (BNE) as the solution to the fixed point equation
\begin{align}\label{Bayes-Nash}
    \sigma^*_{i,t}(h_{i,t})
       = \BR_{i,t}( \sigma^*_{0:t-1},\ \{\sigma^*_{j,t}\}_{j \in V \setminus i}), 
            \quad \forall h_{i,t},
\end{align}
where we have also added the restriction that an equilibrium strategy $\sigma^*_{i,t-1}$ has been played for all times $u<t$. We emphasize that \eqref{Bayes-Nash} needs to be satisfied for all possible histories $h_{i,t}$ and not just for the history realized in a particular game realization. This is necessary because agent $i$ doesn't know the history observed by agent $j$ but rather a probability distribution on histories. Thus, to evaluate the expectation in \eqref{eqn_best_response} agent $i$ needs a representation of the equilibrium strategy for all possible histories $h_{j,t}$.

If all agents play their BNE strategies as defined in \eqref{Bayes-Nash}, $\sigma^*_{i,t}$ becomes optimal in the usual game theoretic sense. There is no strategy that agent $i$ could unilaterally deviate to that provides a higher expected payoff than $\sigma^*_{i,t}$ [cf. \eqref{eqn_best_response}]. In that sense the BNE strategy is the best that agent $i$ can do given other agents' strategies and his locally available information $h_{i,t}$. In the rest of the paper we consider agents playing with respect to the BNE strategy $\sigma^*_{i,t}$ at all times. To simplify future notation define the expectation operator 
\begin{align}\label{eqn_expectation_operator}
    \bbE_{i,t} \big[ \ \cdot\ \big] 
         :=\ \bbE_{\sigma^*_{0:t-1}} \big[ \ \cdot\ |\ h_{i,t} \big],
\end{align}
to represent expectations with respect to the local history $h_{i,t}$ when agents have played the equilibrium strategy $\sigma^*_{0:t-1}$ in all earlier stages of the game. Similarly, we define the conditional probability distribution of agent $i$ at time $t$ given past strategies $\sigma_{0:t-1}^*$ and his information $h_{i,t}$ as 
%
%\begin{equation}\label{eqn_probability_operator}
$\bbP_{i,t}(\cdot) := \bbP_ {\sigma^*_{0:t-1}} \left(\cdot \given h_{i,t}\right)$. 
%\end{equation}
%

Since $u_i(a_i, \{a_{j}\}_{j \in V \setminus i}, \theta)$ is a strictly concave quadratic function of $a_i$ as per \eqref{quadratic_utility_form}, the same is true of the expected utility $\bbE_{i,t}\big[u_i(a_i,\{\sigma_{j,t}\}_{j\in V\setminus i},\theta)\big]$ that we maximize to obtain the best response in \eqref{eqn_best_response}. We can then rewrite \eqref{eqn_best_response} by nulling the derivative of the expected utility with respect to $a_i$. %Performing this operation for the best response to best response strategies
It follows that the fixed point equation in \eqref{Bayes-Nash} can be rewritten as the set of equations
\begin{equation} \label{linear_reply}
   \sigma^*_{i,t}(h_{i,t}) = \sum_{j \in V\setminus\{i\}} 
          \beta_{ij} \bbE_{i,t} [\sigma^*_{j,t}(h_{j,t})] + \delta\,\bbE_{i,t} [ \theta],
\end{equation}
that need to be satisfied for all possible histories $h_{i,t}$ and agents $i$. Our goal is to develop a filter that agents can use to compute their equilibrium actions $a^*_i(t):=\sigma^*_{i,t}(h_{i,t})$ given their observed history $h_{i,t}$. We pursue this in the following section after some remarks.

\begin{remark} \normalfont
It may be of interest to modify the utility in \eqref{quadratic_utility_form} to include more additive terms that are functions of other actions $\{\bba_{j}\}_{j \in V\setminus i}$ and the state of the world $\theta$ but not of the self actions $a_i$. This may change the utility and the expected utility in \eqref{eqn_best_response} but doesn't change the equilibrium strategy in \eqref{Bayes-Nash}. Since these terms do not contain the self action $a_i$, their derivatives are null and do not alter the fixed point equation in \eqref{linear_reply}.
\end{remark}

\begin{remark} \normalfont 
The equilibrium notion in \eqref{Bayes-Nash} is based on the premise of myopic agents that choose actions that optimize payoffs at the present game stage. A more general model is to consider non-myopic agents that consider discounted payoffs of future stages. Non-myopic behavior introduces another layer of strategic reasoning. Forward looking agents would need to take into account the effect of their decisions at each stage of the game on the future path of play knowing that other agents base their future decisions on what they have previously observed. E.g., non-myopic agents might reduce their immediate payoff to harvest information that may result in future gains. Extensions to games with non-myopic agents is beyond the scope of this paper.
\end{remark}

%!TEX root = bqng_arxiv.tex
\section{Propagation of probability distributions} \label{recursion}

According to the model in \eqref{linear_reply}, at each stage of the game agents use the observed history $h_{i,t}$ to estimate the unknown parameter $\theta$ as well as the histories $\{h_{j,t}\}_{j\in V\setminus i}$ observed by other agents. They use the latter and the known BNE strategy $\{\sigma^*_{j,t}(h_{j,t})\}_{j\in V\setminus i}$ to form a belief $\bbP_{i,t}(\{a^*_j(t)\}_{j\in V\setminus i})$ on the actions $\{a^*_j(t)\}_{j\in V\setminus i}$ of other agents which they use to compute their equilibrium action $a^*_j(t)$ at time $t$. Observe that if the vector of private signals $\bbx := [x_1, \ldots, x_N]^T$ is given -- not to the agents but to an outside observer -- the trajectory of the game is completely determined as there are no random decisions. Thus, agent $i$ can form beliefs on the histories $\{h_{j,t}\}_{j\in V\setminus i}$ and actions $\{a^*_j(t)\}_{j\in V\setminus i}$ of other agents if it keeps a local belief $\bbP_{i,t}(\bbx)$ on the vector of private signals $\bbx$. A method to track this probability distribution is derived in this section using a complete induction argument.

Start by assuming that at given time $t$, the posterior distribution $\bbP_{i,t}(\bbx)$ is normal. Recalling the definition of the expectation operator $ \bbE_{i,t} \big[\, \cdot\,\big]$ in \eqref{eqn_expectation_operator}, the mean of this normal distribution is $\bbE_{i,t} \left[\bbx\right]$. Define the corresponding error covariance matrix  $M^i_{\bbx \bbx}(t) \in \reals^{N \times N}$ as 
\begin{equation} \label{error_covariance_xx}
     M^i_{\bbx \bbx}(t) 
          := \bbE_{i,t}\Big[\big(\bbx - \bbE_{i,t} 
               \left[\bbx\right]\big)\big(\bbx - \bbE_{i,t} \left[\bbx\right]\big)^T\Big].
\end{equation}
Although agent $i$'s probability distribution for $\bbx$ is sufficient to describe its belief on the state of the system, subsequent derivations are simpler if we keep an explicit belief on the state of the world $\theta$. Therefore, we also assume that agent $i$'s beliefs on $\theta$ and $\bbx$ are jointly Gaussian given history $h_{i,t}$. The mean of $\theta$ is $\bbE_{i,t} \left[\theta\right] $ and the corresponding variance is
\begin{equation} \label{error_covariance_state}
     M^i_{\theta \theta}(t) 
          := \bbE_{i,t}\Big[\big(\theta - \bbE_{i,t}  \left[\theta\right]\big)
               \big(\theta - \bbE_{i,t}  \left[\theta\right]\big)^T\Big].
\end{equation}
The cross covariance $M^i_{\theta \bbx}(t) \in \reals^{1 \times N}$ between the world state $\theta$ and the private signals $\bbx$ is
\begin{equation}\label{error_covariance_thetax}
M^i_{\theta \bbx}(t) := \bbE_{i,t}\Big[\big(\theta - \bbE_{i,t}  \left[\theta\right]\big)\big(\bbx - \bbE_{i,t} \left[\bbx\right]\big)^T \Big].
\end{equation} 
We further make the stronger assumption that the means of this joint Gaussian distribution can be written as linear combinations of the private signals. In particular, we assume that for some known matrix $L_{i,t} \in \reals^{N \times N}$ and vector $\bbk_{i,t} \in \reals^{N \times 1}$ we can write
\begin{align} \label{linear_estimates_first_one}
   \bbE_{i,t} \left[\bbx\right] = L_{i,t} \bbx, \qquad
   \bbE_{i,t} \left[\theta\right] = \bbk_{i,t}^T  \bbx.
\end{align}
Observe that the assumption in \eqref{linear_estimates_first_one} is not that the estimates $\bbE_{i,t} \left[\bbx\right]$  and $\bbE_{i,t}\left[\theta\right] $ are computed as linear combinations of the private signals $\bbx$ -- indeed, $\bbx$ is not known by agent $i$ in general. The assumption is that from the perspective of an external observer the actual computations that agents do are equivalent to the linear transformations in \eqref{linear_estimates_first_one}.

Under the complete induction hypothesis of Gaussian posterior beliefs at time $t$ with expectations as in \eqref{linear_estimates_first_one}, we show that agents play according to linear equilibrium strategies of the form 
\begin{equation} \label{linear_action}
    \sigma^*_{i,t}(h_{i,t}) = \bbv_{i,t}^T \bbE_{i,t}[\bbx],  % \quad\forall i \in V
\end{equation}
for some action coefficients $\bbv_{i,t} \in \reals^{N\times 1}$ that vary across agents but are independent of the observed history $h_{i,t}$. These can be found by solving a system of linear equations. We do this in the following lemma.

%%%%%%%%%%%%%%%%%%%%%%%%%%%%%%%%%%%%%%%%%%%%%%%%%%%%%%%%%%%%%%%%%%%%
%   L   E   M   M   A   %%%%%%%%%%%%%%%%%%%%%%%%%%%%%%%%%%%%%%%%%%%%
%%%%%%%%%%%%%%%%%%%%%%%%%%%%%%%%%%%%%%%%%%%%%%%%%%%%%%%%%%%%%%%%%%%%
%
\begin{lemma} \label{linear_equilibrium_at_time_t}
Consider a Bayesian game with quadratic utility as in \eqref{quadratic_utility_form}. Suppose that for all agents $i$, the joint posterior beliefs $\bbP_{i,t}([\theta, \bbx^T])$ on the state of the world $\theta$ and the private signals $\bbx$ given the local history $h_{i,t}$ at time $t$ are Gaussian with means expressed as the linear combinations of private signals in \eqref{linear_estimates_first_one} for some known vectors $\bbk_{i,t}$ and matrices $L_{i,t}$. Define the aggregate vector $\bbk_t :=  [\bbk_{1,t}^T, \ldots, \bbk_{N,t}^T]^T \in \reals^{N^2 \times 1} $ stacking the state estimation weights of all agents
%, the block diagonal matrix $\Lambda_t := {\rm diag}(\{L_{j,t}^T\}_{j \in V})\in \reals^{N^2\times N^2}$ with $N \times N$ diagonal blocks $((\Lambda_t))_{ii} = L_{i,t}^T$ 
and the block matrix $L_t \in \reals^{N^2\times N^2}$ with $N\times N$ diagonal blocks $((L_t))_{ii} = L_{i,t}^T$ and off diagonal blocks $((L_t))_{ij} = -\beta_{ij} L_{i,t}^T L_{j,t}^T$,
\begin{equation} \label{multipliers_matrix}
 L_t\hspace{-1mm}:=\hspace{-1mm} \begin{psmallmatrix}
  L_{1,t}^T & -\beta_{12}L_{1,t}^T L_{2,t}^T & \ldots & -\beta_{1N}L_{1,t}^T  L_{N,t}^T \\
   -\beta_{21} L_{2,t}^T  L_{1,t}^T &L_{2,t}^T&\ldots & -\beta_{2N}L_{2,t}^T  L_{N,t}^T \\
  \vdots & \cdots & \ddots &\vdots \\
   -\beta_{N-11} L_{N-1,t}^T  L_{1,t}^T & \cdots &\qquad L_{N-1,t}^T &  -\beta_{N-1N}L_{N-1,t}^T L_{N,t}^T  \\
   -\beta_{N1} L_{N,t}^T  L_{1,t}^T & \cdots & -\beta_{NN-1}L_{N,t}^T  L_{N-1,t}^T &   L_{N,t}^T
 \end{psmallmatrix} .
 \end{equation}
If there exists a linear equilibrium strategy as in \eqref{linear_action}, the action coefficients $\bbv_t:= [\bbv_{1,t}^T, \ldots, \bbv_{N,t}^T]^T \in \reals^{N^2} $ can be obtained by solving the system of linear equations
\begin{equation} \label{BNE_matrix_form}
    L_t  \bbv_t = \delta \bbk_t.
\end{equation} \end{lemma}

%%%%%%%%%%%%%%%%%%%%%%%%%%%%%%%%%%%%%%%%%%%%%%%%%%%%%%%%%%%%%%%%%%%%
%   P   R   O   O   F   %%%%%%%%%%%%%%%%%%%%%%%%%%%%%%%%%%%%%%%%%%%%
%%%%%%%%%%%%%%%%%%%%%%%%%%%%%%%%%%%%%%%%%%%%%%%%%%%%%%%%%%%%%%%%%%%%
%
\begin{proof} We hypothesize that agents play according to a linear equilibrium strategy as in \eqref{linear_action}. Substituting this candidate strategy into the equilibrium equations in \eqref{linear_reply} yields
\begin{equation} \label{BNE_equations_t}
    \bbv_{i,t}^T \bbE_{i,t}[\bbx] 
         = \sum_{j \in V\setminus\{i\}} 
               \beta_{ij} \bbE_{i,t} \Big[\bbv_{j,t}^T \bbE_{j,t}[\bbx]\Big] 
                + \delta \,  \bbE_{i,t}[\theta].
\end{equation}
The summation in \eqref{BNE_equations_t} includes the expectations $\bbE_{i,t} \big[\bbE_{j,t}[\bbx]\big]$ of agent $i$ on the private signals' estimate of agent $j$. As per the induction hypothesis in \eqref{linear_estimates_first_one}, we have that the inner expectations can be written as $\bbE_{j,t}[\bbx] = L_{j,t}\bbx$. Using this fact, agent $i$'s expectation of agent $j$'s estimate of private signals becomes
\begin{equation} \label{i_j_estimate}
    \bbE_{i,t} \Big[\bbE_{j,t}[\bbx]\Big] = L_{j,t} \bbE_{i,t} [\bbx].
\end{equation}
Substituting \eqref{i_j_estimate} and the estimate induction hypotheses in \eqref{linear_estimates_first_one} for the corresponding terms in \eqref{BNE_equations_t} and \eqref{i_j_estimate}, and reordering terms yield the set of equations
\begin{equation}\label{BNE_equations_t_v1}
    \bbv_{i,t}^T L_{i,t} \bbx 
        = \sum_{j \in V\setminus\{i\}} 
             \beta_{ij}\bbv_{j,t}^T L_{j,t} L_{i,t} \bbx 
             + \delta \,  \bbk_{i,t}^T \bbx,
\end{equation}
At this point we recall that the equilibrium equations in \eqref{linear_reply} are true for all possible histories $h_{i,t}$. Therefore, the equilibrium equations in \eqref{BNE_equations_t_v1}, which are derived from \eqref{linear_reply}, have to hold irrespectively of the history's realization. This in turn means that they will be true for all possible values of $\bbx$. This can be ensured by equating the coefficients that multiply each component of $\bbx$ in \eqref{BNE_equations_t_v1} thereby yielding the relationships
\begin{equation}\label{BNE_equations_t_v2}
     L_{i,t}^T \bbv_{i,t} 
         = \sum_{j \in V\setminus\{i\}} \beta_{ij} L_{i,t}^T L^T_{j,t}\bbv_{j,t} 
           + \delta \,  \bbk_{i,t},
\end{equation}
that need to hold true for all agents $i$. The result in \eqref{BNE_matrix_form} is just a restatement of \eqref{BNE_equations_t_v2} with the latter corresponding to the $i$-th block of the relationship in \eqref{BNE_matrix_form}. \end{proof}

%%%%%%%%%%%%%%%%%%%%%%%%%%%%%%%%%%%%%%%%%%%%%%%%%%%%%%%%%%%%%%%%%%%%
%   M   A   I   N       M   A   T   T   E   R   %%%%%%%%%%%%%%%%%%%%
%%%%%%%%%%%%%%%%%%%%%%%%%%%%%%%%%%%%%%%%%%%%%%%%%%%%%%%%%%%%%%%%%%%%
%
Lemma \ref{linear_equilibrium_at_time_t} provides a mechanism to determine the strategy profiles $\sigma^*_{i,t}(\cdot)$ of all agents through the computation of the action vectors $\bbv_{i,t}$ as a block of the vector $\bbv_t$ that solves \eqref{BNE_matrix_form}. We emphasize that the value of the weight vector $\bbv_t$ in \eqref{BNE_matrix_form} does not depend on the realization of private signals $\bbx$. This is as it should because the postulated equilibrium strategy in \eqref{linear_action} assumes the action weights $\bbv_{i,t}$ are independent of the observed history.  A consequence of this fact is that the action coefficients $\{\bbv_{i,t}\}_{i \in V}$ of all agents can be determined locally by all agents as long as the matrices $L_{i,t}$ and vectors $\bbv_{i,t}$ are common knowledge. The equilibrium actions $a^*_{i}(t)$, however, do depend on the observed history because to determine the action $a^*_{i}(t) = \sigma^*_{i,t}(h_{i,t}) = \bbv_{i,t}^T \bbE_{i,t}[\bbx]$ we multiply $\bbv_{i,t}^T$ by the expectation $\bbE_{i,t}[\bbx]$ associated with the actual observed history $h_{i,t}$. See Section \ref{QNG_filter} for details.

At time $t$ agent $i$ computes its action vector $\bbv_{i,t}$ which it uses to select the equilibrium action $a^*_{i}(t)=\bbv_{i,t}^T \bbE_{i,t}[\bbx]$ as per \eqref{linear_action}. Since we have also hypothesized that $\bbE_{i,t} \left[\bbx\right] = L_{i,t} \bbx$, as per \eqref{linear_estimates_first_one} the action of agent $i$ at time $t$ is given by
\begin{equation} \label{linear_action_assumption}
    a_{i}(t) = \bbv_{i,t}^T L_{i,t} \bbx .
\end{equation}
We emphasize that as in \eqref{linear_estimates_first_one} the expression in \eqref{linear_action_assumption} is not the computation made by agent $i$ but an equivalent computation from the perspective of an external omniscient observer.

The actions $\bba_{n(i)}(t):=[a_{j_{i,1}}(t),\ldots,a_{j_{i,d(i)}}(t)]^T \in \reals^{d(i) \times 1}$ of neighboring agents $j\in n(i)$ become part of the observed history $h_{i,t+1}$ of agent $i$ at time $t+1$ [cf. \eqref{eqn_history_definition}]. The important consequence of \eqref{linear_action_assumption} is that these observations are a linear combination of private signals $\bbx$. In particular, by defining the matrix $H_{i,t}^T:= [\bbv_{j_{i,1},t}^T L_{j_{i,1},t}; \ldots; \bbv_{j_{i,d(i)},t}^T L_{j_{i,d(i)},t}] \in \reals^{d(i) \times N}$ we can write
\begin{align} \label{observation_matrix}
   \bba_{n(i)}(t) = H_{i,t}^T \bbx 
                 := \left(\begin{array}{c}
                       \bbv_{j_{i,1},t}^T L_{j_{i,1},t} \\
                       \vdots\\
                        \bbv_{j_{i,d(i)},t}^T L_{j_{i,d(i)},t}
                     \end{array}\right)
                     \bbx .
\end{align}
Agent $i$'s belief of $\bbx$ at time $t$ is normally distributed; moreover, when we go from time $t$ to time $t+1$, agent $i$ observes a linear combination, $\bba_{n(i)}(t) = H_{i,t}^T \bbx$, of private signals. Thus, the propagation of the probability distribution when the history $h_{i,t+1}$ incorporates the actions $\bba_{n(i)}(t)$ is a simple sequential LMMSE estimation problem \cite[Ch. 12]{Kay}. In particular, the joint posterior distribution of $\bbx$ and $\theta$ given $h_{i,t+1}$ remains Gaussian and the expectations $\bbE_{i,t+1} \left[\bbx\right]$ and $\bbE_{i,t+1} \left[\theta\right]$ remain linear combinations of private signals $\bbx$ as in \eqref{linear_estimates_first_one} for some matrix $L_{i,t+1}$ and vector $\bbk_{i,t+1}$ which we compute explicitly in the following lemma.

%%%%%%%%%%%%%%%%%%%%%%%%%%%%%%%%%%%%%%%%%%%%%%%%%%%%%%%%%%%%%%%%%%%%
%   L   E   M   M   A   %%%%%%%%%%%%%%%%%%%%%%%%%%%%%%%%%%%%%%%%%%%%
%%%%%%%%%%%%%%%%%%%%%%%%%%%%%%%%%%%%%%%%%%%%%%%%%%%%%%%%%%%%%%%%%%%%
%
\begin{lemma} \label{rational_belief_updates_theorem_L}
Consider a Bayesian game with quadratic utility as in \eqref{quadratic_utility_form} and the same assumptions and definitions of Lemma \ref{linear_equilibrium_at_time_t}. Further define the observation matrix $H_{i,t}^T:= [\bbv_{j_{i,1},t}^T L_{j_{i,1},t}; \ldots; \bbv_{j_{i,d(i)},t}^T L_{j_{i,d(i)},t}] \in \reals^{d(i) \times N}$ as in \eqref{observation_matrix} and the LMMSE gains
\begin{IEEEeqnarray}{rCl}
K^i_\bbx(t) & :=  & M^i_{\bbx \bbx} (t) H_{i,t}\big(H_{i,t}^T  M^i_{\bbx \bbx} (t) H_{i,t}\big)^{-1}, \label{eqn_lmmse_gain_x} \\
K^i_\theta(t) & := & M^i_{\theta \bbx} (t) H_{i,t} \big(H_{i,t}^T  M^i_{\bbx \bbx} (t) H_{i,t}\big)^{-1}, \label{eqn_lmmse_gain_theta}
\end{IEEEeqnarray}
and assume that agents play the linear equilibrium strategy in \eqref{linear_action}. 
Then, the beliefs $\bbP_{i,t+1}([\theta, \bbx^T])$ after observing neighboring actions at time $t$ are Gaussian with means that can be expressed as the linear combination of private signals 
\begin{align} \label{linear_estimates_next_time}
   \bbE_{i,t+1} \left[\bbx\right] = L_{i,t+1} \bbx, \qquad
   \bbE_{i,t+1} \left[\theta\right] = \bbk_{i,t+1}^T  \bbx,
\end{align}
where the matrix $L_{i,t+1}$ and vector $\bbk_{i,t+1}$ are given by
\begin{align}
L_{i,t+1} &= L_{i,t} +  K^i_\bbx(t) \Big(H_{i,t}^T -  H_{i,t}^T L_{i,t}\Big) \label{weights_recursion_x},\\
\bbk_{i,t+1}^T &= \bbk_{i,t}^T + K^i_\theta(t) \Big(H_{i,t}^T - H_{i,t}^T L_{i,t}\Big)\label{weights_recursion_theta}.
\end{align} 
The posterior covariance matrix $M^i_{\bbx \bbx} (t+1)$ for the private signals $\bbx$ the variance $M^i_{\theta \theta} (t+1)$ of the state $\theta$ and the cross covariance $M^i_{\theta \bbx} (t+1)$ are further given by
\begin{align}
M^i_{\bbx \bbx} (t+1) =& M^i_{\bbx \bbx} (t)- K^i_\bbx(t) H_{i,t}^T M^i_{\bbx \bbx} (t)\label{LMMSE_updates_covariance}, \\
M^i_{\theta \theta} (t+1) =& M^i_{\theta \theta} (t) - K_\theta^i(t)^T H_{i,t}^T M^i_{\bbx \theta} (t), \label{theta_covariance_estimate} \\
M^i_{\theta \bbx} (t+1) =& M^i_{\theta \bbx} (t) -  K^i_\theta(t) H_{i,t}^T M^i_{\bbx \bbx} (t).\label{theta_x_covariance_estimate}
\end{align}
\end{lemma}

%%%%%%%%%%%%%%%%%%%%%%%%%%%%%%%%%%%%%%%%%%%%%%%%%%%%%%%%%%%%%%%%%%%%
%   P   R   O   O   F   %%%%%%%%%%%%%%%%%%%%%%%%%%%%%%%%%%%%%%%%%%%%
%%%%%%%%%%%%%%%%%%%%%%%%%%%%%%%%%%%%%%%%%%%%%%%%%%%%%%%%%%%%%%%%%%%%
%
\begin{proof} 
Since observations of $i$, $a_{n(i)}(t)$, are linear combinations of private signals $\bbx$ which are normally distributed, observations of $i$ are also normally distributed from the perspective of $i$. Furthermore, by assumption \eqref{linear_estimates_first_one}, the prior distribution $\bbP_{i,t}(\bbx)$ is Gaussian. Hence, the posterior distribution, $\bbP_{i,t+1}(\bbx)$, is also Gaussian. Specifically, the mean of the posterior distribution corresponds to the LMMSE estimator with gain matrix  $K^i_\bbx(t) =  M^i_{\bbx \bbx} (t) H_{i,t}\big(H_{i,t}^T  M^i_{\bbx \bbx} (t) H_{i,t}\big)^{-1}$; that is,
\begin{align} 
\bbE_{i,t+1} [\bbx] =&  \bbE_{i,t} \left[\bbx\right] +  K^i_\bbx(t)  \big(\bba_{n(i)}(t)  -\bbE_{i,t} [\bba_{n(i)}(t)]\big).  \label{LMMSE_updates_mean}
\end{align}

Because $\theta$ and $\bbx$ are jointly Gaussian at time $t$, $\theta$ and $\bba_{n(i)}(t)$ are also jointly Gaussian. Therefore, the posterior distribution $\bbP_{i,t+1}(\theta)$ is also Gaussian. Consequently, the Bayesian estimate of $\theta$ is given by a sequential LMMSE estimator with gain matrix  $K^i_\theta(t) = M^i_{\theta \bbx} (t) H_{i,t} \big(H_{i,t}^T  M^i_{\bbx \bbx} (t) H_{i,t}\big)^{-1}$,
\begin{align} 
\bbE_{i,t+1} \left[\theta\right] =& \bbE_{i,t} \left[\theta\right] + K^i_\theta(t)  \left( \bba_{n(i)}(t) - \bbE_{i,t}\left[\bba_{n(i)}(t)\right]\right) \label{theta_mean_estimate}.
\end{align}
Given the linear observation model in \eqref{observation_matrix}, agent $i$'s estimate of his observations at time $t$ is given by $\bbE_{i,t} (\bba_{n(i)}(t) )=  H_{i,t}^T \bbE_{i,t}[\bbx]$. Substituting \eqref{linear_estimates_first_one} for the mean estimates at time $t$ in \eqref{LMMSE_updates_mean} and \eqref{theta_mean_estimate}, we obtain
\begin{align}\bbE_{i,t+1}\left[\bbx\right] &= L_{i,t} \bbx +  K^i_\bbx(t) \left(H_{i,t}^T \bbx-  H_{i,t}^T L_{i,t} \bbx\right),\label{LMMSE_updates_mean_10}
\\
\bbE_{i,t+1}\left[\theta \right]  &= \bbk_{i,t}^T \bbx + K^i_\theta(t) \left(H_{i,t}^T \bbx - H_{i,t}^T L_{i,t} \bbx \right)\label{theta_mean_estimate_10}
.
\end{align}
Grouping the terms that multiply $\bbx$ on the right hand side of the two equations, we observe that $\bbE_{i,t+1}\left[\bbx\right]  = L_{i,t+1} \bbx$ and $\bbE_{i,t+1}\left[\theta\right]  = \bbk_{i,t+1}^T \bbx$ where $L_{i,t+1}$ and $\bbk_{i,t+1}$ are as defined in \eqref{weights_recursion_x} and \eqref{weights_recursion_theta}. Similarly,  the updates for error covariance matrices are as given in \eqref{LMMSE_updates_covariance}--\eqref{theta_x_covariance_estimate} following standard LMMSE updates  \cite[Ch. 12]{Kay}. 
\end{proof}

%%%%%%%%%%%%%%%%%%%%%%%%%%%%%%%%%%%%%%%%%%%%%%%%%%%%%%%%%%%%%%%%%%%%
%   M   A   I   N       M   A   T   T   E   R   %%%%%%%%%%%%%%%%%%%%
%%%%%%%%%%%%%%%%%%%%%%%%%%%%%%%%%%%%%%%%%%%%%%%%%%%%%%%%%%%%%%%%%%%%
%
In the repeated game we are considering, agents determine optimal actions given available information and determine the information that is revealed by neighboring actions. These questions are respectively answered by lemmas \ref{linear_equilibrium_at_time_t} and \ref{rational_belief_updates_theorem_L} under the inductive hypotheses of Gaussian beliefs and linear estimates as per \eqref{linear_estimates_first_one}. The answer provided by Lemma \ref{rational_belief_updates_theorem_L} also shows that the inductive hypotheses hold true at time $t+1$ and provides an explicit recursion to propagate the mean and variance of the beliefs posterior to the observation of neighboring actions. This permits closing the inductive loop to establish the following theorem for recursive computation of BNE of repeated games with quadratic payoffs.

%%%%%%%%%%%%%%%%%%%%%%%%%%%%%%%%%%%%%%%%%%%%%%%%%%%%%%%%%%%%%%%%%%%%
%   T   H   E   O   R   E   M   %%%%%%%%%%%%%%%%%%%%%%%%%%%%%%%%%%%%
%%%%%%%%%%%%%%%%%%%%%%%%%%%%%%%%%%%%%%%%%%%%%%%%%%%%%%%%%%%%%%%%%%%%
%
\begin{theorem} \label{rational_updates_theorem}
Consider a repeated Bayesian game with the quadratic utility function in \eqref{quadratic_utility_form} and assume that linear strategies $\sigma^*_{i,t}(h_{i,t}) = \bbv_{i,t}^T \bbE_{i,t}[\bbx]$ as in \eqref{linear_action} exist for all times $t$. Then, the action coefficients $\bbv_{i,t}$ can be computed by solving the system of linear equations in \eqref{BNE_matrix_form} with $\bbv_t:= [\bbv_{1,t}^T, \ldots, \bbv_{N,t}^T]^T$, $\bbk_t :=  [\bbk_{1,t}^T, \ldots, \bbk_{N,t}^T]^T$ and $L_t$ as in \eqref{multipliers_matrix}. The matrices $L_{i,t}$ and the vectors $\bbk_{i,t}$ are computed by recursive application of \eqref{eqn_lmmse_gain_x}-\eqref{eqn_lmmse_gain_theta} and \eqref{weights_recursion_x}-\eqref{theta_x_covariance_estimate} with initial values %\footnote{Recall that $\bbe_i$ is the $i$th element of the standard orthonormal basis of $\reals^N$ and $\bar\bbe_i:=\bbone-\bbe_i$ denotes an all-one vector with the $i$th component nulled.}
\begin{align} \label{initial_weights}
   L_{i,0}    = \bbone\bbe_i^T, \qquad
   \bbk_{i,0} = \bbe_i. 
\end{align}
The initial covariance matrix $M^i_{\bbx \bbx} (0)$, initial variance $M^i_{\theta \theta} (0)$, and initial cross covariance $M^i_{\theta \bbx} (0)$ are given by 
\begin{align} \label{initial_variances}
   & M^i_{\bbx \bbx} (0)    
       = {\rm diag}(\bar\bbe_i){\rm diag} (\bbc) + \bar\bbe_i\bar\bbe_i^T c_i, \hspace{5mm}  
    M^i_{\theta \theta}(0) = c_{i},           \hspace{5mm}
    M^i_{\theta \bbx}(0)   = c_i\bar\bbe_i^T. 
\end{align} \end{theorem}

%%%%%%%%%%%%%%%%%%%%%%%%%%%%%%%%%%%%%%%%%%%%%%%%%%%%%%%%%%%%%%%%%%%%
%   P   R   O   O   F   %%%%%%%%%%%%%%%%%%%%%%%%%%%%%%%%%%%%%%%%%%%%
%%%%%%%%%%%%%%%%%%%%%%%%%%%%%%%%%%%%%%%%%%%%%%%%%%%%%%%%%%%%%%%%%%%%
%
\begin{proof} See Appendix \ref{Theorem_1_proof}.
\end{proof}

%%%%%%%%%%%%%%%%%%%%%%%%%%%%%%%%%%%%%%%%%%%%%%%%%%%%%%%%%%%%%%%%%%%%
%   M   A   I   N       M   A   T   T   E   R   %%%%%%%%%%%%%%%%%%%%
%%%%%%%%%%%%%%%%%%%%%%%%%%%%%%%%%%%%%%%%%%%%%%%%%%%%%%%%%%%%%%%%%%%%
%
%\blue{Please clean up the Theorem's comment and the setting up of the following sections.}

According to Theorem \ref{rational_updates_theorem}, the beliefs on $\theta$ and $\bbx$ remain Gaussian for all agents and all times when agents play according to a linear equilibrium strategy as in \eqref{linear_action} at each stage. Theorem \ref{rational_updates_theorem} also provides a recursive mechanism to compute the coefficients $\bbv_{i,t}$ of the linear BNE strategies $\sigma^*_{i,t}(h_{i,t}) = \bbv_{i,t}^T \bbE_{i,t}[\bbx]$ and the coefficients $L_{i,t}$ and $\bbk_{i,t}$ that determine the LMMSE estimates as per \eqref{linear_estimates_first_one}. However, these latter expressions cannot be used by agent $i$ to calculate estimates $\bbE_{i,t} \left[\bbx\right]$ and $\bbE_{i,t} \left[\theta\right]$ unless the private signals $\bbx$ are exactly known, which will absolve agent $i$ from responsibility of the estimation process entirely. Since the BNE action $a^*_{i}(t)=\sigma^*_{i,t}(h_{i,t}) = \bbv_{i,t}^T \bbE_{i,t}[\bbx]$ depends on having the observed private signal estimate $\bbE_{i,t}[\bbx]$ available, Theorem \ref{rational_updates_theorem} does not provide a way of computing the optimal action either. This mismatch can be solved by writing the LMMSE updates in a different form as we show in the next section after the following remark.

\begin{remark} \normalfont
Results in this paper assume the system of linear equations in \eqref{BNE_matrix_form} has a unique solution. If the solution is not unique, a prior agreement is necessary for agents to play consistent strategies. E.g., agents could agree beforehand to select the vector $\bbv_t$ with minimum Euclidean norm. If \eqref{BNE_matrix_form} does not have a solution, it means that the equilibrium strategies of the form in \eqref{linear_action_assumption} do not exist. A sufficient condition for this {\it not} to happen is to have a strictly diagonally dominant utility function which in explicit terms we write $\sum_{j\in V\setminus\{i\}} |\beta_{ij}| < 1$. In this case Gershgorin's Theorem implies that $L_t$ is full rank because it has no null eigenvalues. Laxer conditions to guarantee existence of linear equilibria as in \eqref{linear_action_assumption} can be found in, e.g., \cite{ui2009bayesian,radner1962team}. In all of our numerical experiments solutions to \eqref{BNE_matrix_form} exist and are unique.
\end{remark}
 
%\input{existence_and_uniqueness.tex}

%!TEX root = bqng_arxiv.tex

\begin{figure*} 
\centering
%!TEX root = ../bqng.tex

\def \lbra {[} \def \rbra {]}

\tikzstyle{empty} = [rectangle, draw=black, inner sep=0pt, minimum size=0cm, fill=blue!10]
\tikzstyle{box} = [rectangle, draw=black, inner sep=0pt, minimum size=0.9cm, fill=blue!10]
\tikzstyle{hugebox} = [rectangle, draw=black, inner sep=0pt, minimum width = 8.2cm, minimum height = 6.6cm,dashed]
\tikzstyle{agent} = [circle, draw=black, inner sep=0pt, minimum size=0.7cm, fill=red!10]
\tikzstyle{arrow} = [->]
{\begin{tikzpicture}[scale=1.2]
\path    ( 0,    0)   node (0) [rectangle, draw=black, inner sep=0pt, minimum size=0cm, label= left: $\bba_{n(i)}(t)$]  {}
      ++  (1,   0)     node (1) [agent]  {$\sum$}
      ++ ( 2,  0 )       node (2) [box]  {$K_\bbx^i$}
      ++ ( 2,  0)   node (3) [agent]  {$\sum$}
  ++ ( 1,   0) node (4) [circle, draw=black,inner sep=0pt, minimum size=0cm, label=above: $\bbE_{i,t}\lbra\bbx\rbra$] {}
        ++ ( 2,   0)    node (7) [box]  {$\bbv_{i,t}^T$}
        ++ ( 1,   0)     node (8) [rectangle, draw=black, inner sep=0pt, minimum size=0cm, label=right:$a_{i}(t)$]  {}
              ++ ( -3,   -2)    node (5) [empty]  {}
      ++ ( -1,   0)    node (6) [empty]  {}
++ ( -2,   0)     node (9) [box]{$-H_{i,t}^T$}
         ++ ( -2,   0)     node (10) [rectangle, draw=black, inner sep=0pt, minimum size=0cm, label=below:$-\bbE_{i,t}\lbra \bba_{n(i)}(t)\rbra$] {}
         ++ (2, 3) node (11)  [rectangle, draw=black, inner sep=0pt, minimum size=0cm, label=above:$M_{\bbx\bbx}^i(t)$]{}
         ++ (5, 0) node (12)[circle, draw=black,inner sep=0pt, minimum size=0cm, label=above: $\{\bbk_{j,t}\}_{j \in V}$]{}
         ++ (0, -2) node (13) [circle, draw=black,inner sep=0pt, minimum size=0cm, label=below: $\{L_{j,t}\}_{j \in V}$]{}
         ++ (-10, 1) node (14) [box]{$H_{i,t}^T$}
	++ (-1, 0) node (15)[circle, draw=black,inner sep=0pt, minimum size=0cm, label=above: $\bbx$] {}
	++ (1, 1) node (16)[circle, draw=black,inner sep=0pt, minimum size=0cm, label=above: $\{\bbv_{j,t}\}_{j \in n(i)}$] {}
	++ (0, -2) node (17)[circle, draw=black,inner sep=0pt, minimum size=0cm, label=below: $\{L_{j,t}\}_{j \in n(i)}$] {}
	++ (2, 1) node (18)[empty] {}
	++ (-1.1, 0) node (19)[empty] {}
	++ (+3, 0) node (20)[empty] {}
	++ (+0, 2.5) node (21)[empty] {}
	++ (+1, 0) node (22)[box] {$K_\theta^i$}
	++ (+2, 0) node (23)[agent] {$\sum$}
	++ (+1, 0) node (24)[circle, draw=black,inner sep=0pt, minimum size=0cm, label=above: $\bbE_{i,t}\lbra\theta\rbra$]{}
	++ (0, -1) node (25)[empty]{}
	++ (-1, 0) node (26)[empty]{}
	++ (3, 1) node (27)[empty]{}
%	++ (-0.5, 1) node (28)[empty]{}
%	++ (-5, 0) node (29)[empty]{}
%	++ (0, -5) node (30)[empty]{}
%	++ (5, 0) node (31)[empty]{}
	++ (-4.4, -2.2) node (32)[hugebox]{};
	
\draw [arrow, black] (0)  to (1);
\draw [arrow, black] (1)  to (2);
\draw [arrow, black] (2)  to (3);
\draw [line width = 0.2mm, -] (3)  to (4);
\draw [line width = 0.2mm, -] (4)  to (5);
\draw [line width = 0.2mm, -] (5)  to (6);
\draw [arrow, black] (6)  to (3);
\draw [arrow, black] (4)  to (7);
\draw [arrow, black] (7)  to (8);
\draw [arrow, black] (6)  to (9);
\draw [line width = 0.2mm, -] (9)  to (10);
\draw [arrow, black] (10)  to (1);
\draw [arrow, black] (11)  to (2);
\draw [arrow, black] (12)  to (7);
\draw [arrow, black] (13)  to (7);
\draw [line width = 0.2mm, -] (18)  to (0);
\draw [arrow, black] (15)  to (14);
\draw [arrow, black] (16)  to (14);
\draw [arrow, black] (17)  to (14);
\draw [arrow, black] (17)  to (14);
\draw [arrow, black] (14)  to (19);
\draw [line width = 0.2mm, -] (20)  to (21);
\draw [arrow, black] (21)  to (22);
\draw [arrow, black] (22)  to (23);
\draw [line width = 0.2mm, -] (23)  to (24);
\draw [line width = 0.2mm, -] (24)  to (25);
\draw [line width = 0.2mm, -] (25)  to (26);
\draw [arrow, black] (26)  to (23);
\draw [arrow, black] (24)  to (27);
%\draw [line width = 0.2mm, -] (28)  to (29);
%\draw [line width = 0.2mm, -] (29)  to (30);
%\draw [line width = 0.2mm, -] (30)  to (31);
%\draw [line width = 0.2mm, -] (31)  to (28);
\end{tikzpicture}}\vspace{-3mm}
\caption{Quadratic Network Game (QNG) filter at agent $i$. There are two types of blocks, circle and rectangle. Arrows coming into the circle block are
summed. The arrow that goes into a rectangle block is multiplied by the coefficient written inside the block. Inside the dashed box agent $i$'s mean estimate updates on $\bbx$ and $\theta$ are illustrated (cf. (\ref{LMMSE_updates_mean_filter}) and (\ref{LMMSE_updates_state_mean_filter})). The gain coefficients for the mean updates are fed from LMMSE block in
Fig. 2. The observation matrix $H_{i,t}$ is fed from the game block in Fig. 2. Agent $i$ multiplies his mean estimate on $\bbx$ at time $t$ with action coefficient $\bbv_{i,t}$, which is fed from game block in Fig. 2, to obtain $a_{i}(t)$. The mean estimates $\bbE_{i,t}[\bbx]$ and $a_{i}(t)$ can
only be calculated by agent $i$.}\vspace{-3mm}
\label{feedback_diagram}
\end{figure*}
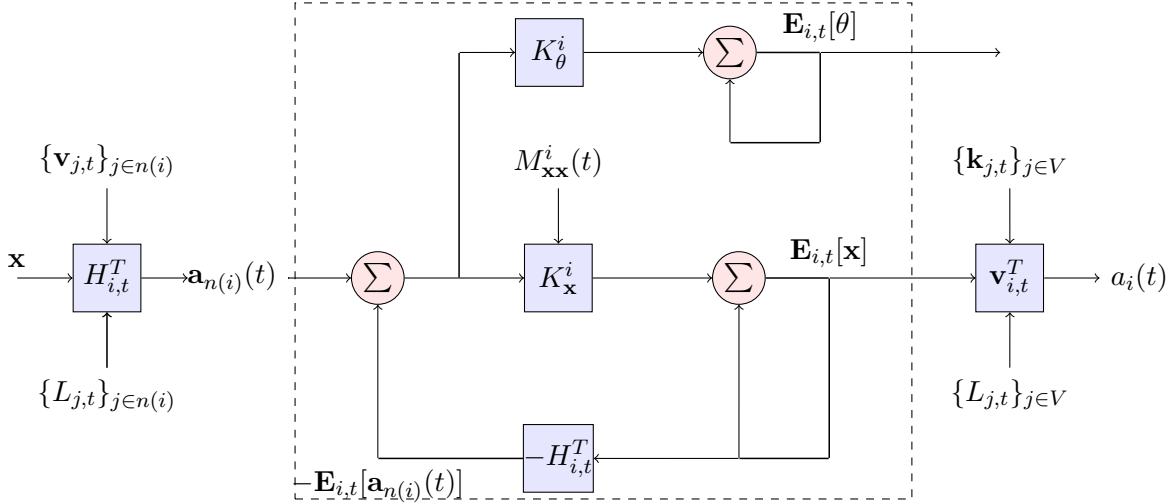

\section{Quadratic Network Game Filter} \label{QNG_filter}

To compute and play BNE strategies each node runs the quadratic network game (QNG) filter that we derive in this section. Since agent $i$ cannot use \eqref{linear_estimates_first_one}, we need an alternative means of computing estimates $\bbE_{i,t} \left[\bbx\right]$ and $\bbE_{i,t} \left[\theta\right]$. To do this refer to the transformation of \eqref{LMMSE_updates_mean} and \eqref{theta_mean_estimate} into \eqref{LMMSE_updates_mean_10} and \eqref{theta_mean_estimate_10} in the proof of Lemma \ref{rational_belief_updates_theorem_L}. In this transformation we substitute the observed neighboring actions $\bba_{n(i)}(t)$ for their model $\bba_{n(i)}(t) = H_{i,t}^T \bbx$ and write the expectation of these actions as
$H_{i,t}^T \bbE_{i,t}[\bbx]$ with the further substitution $\bbE_{i,t} \left[\bbx\right] = L_{i,t} \bbx$. As a result we can rewrite \eqref{LMMSE_updates_mean} and \eqref{theta_mean_estimate} as
\begin{alignat}{4} 
   &\bbE_{i,t+1} [\bbx]   
      =\ && \bbE_{i,t} \left[\bbx\right] 
         &&+ K^i_\bbx(t)&&\big(\bba_{n(i)}(t)-H_{i,t}^T\bbE_{i,t}[\bbx]\big),
               \label{LMMSE_updates_mean_filter} \\
   &\bbE_{i,t+1} [\theta] 
      =\ && \bbE_{i,t} \left[\theta\right] 
         &&+  K^i_\theta(t)  &&\big(\bba_{n(i)}(t)- H_{i,t}^T\bbE_{i,t}[\bbx]\big).  
               \label{LMMSE_updates_state_mean_filter}
\end{alignat}
The updates in \eqref{LMMSE_updates_mean_filter} and \eqref{LMMSE_updates_state_mean_filter} can be implemented locally by agent $i$ since they depend on the previous values $\bbE_{i,t}[\bbx]$ and $\bbE_{i,t}[\theta]$ of the LMMSE estimates, and the observed neighboring actions $\bba_{n(i)}(t)$. They can be combined with the coefficient recursions in \eqref{BNE_matrix_form}, \eqref{eqn_lmmse_gain_x}-\eqref{eqn_lmmse_gain_theta}, and \eqref{weights_recursion_x}-\eqref{theta_x_covariance_estimate} as well as with the BNE strategy expression in \eqref{linear_action} to recursively compute the equilibrium actions $a^*_{i}(t)$ given the observed history $h_{i,t}$.

The updates in \eqref{linear_action}, \eqref{BNE_matrix_form}, \eqref{eqn_lmmse_gain_x}-\eqref{eqn_lmmse_gain_theta}, \eqref{weights_recursion_x}-\eqref{theta_x_covariance_estimate}, and \eqref{LMMSE_updates_mean_filter}-\eqref{LMMSE_updates_state_mean_filter} form the QNG filter. In the QNG filter agent $i$ performs a full network simulation in which it maintains a belief $\bbP_{i,t}([\theta, \bbx^T])$ on the state of the world $\theta$ and the private signals $\bbx$ of all agents. This implies performing the coefficient updates \eqref{BNE_matrix_form}, \eqref{eqn_lmmse_gain_x}-\eqref{eqn_lmmse_gain_theta}, \eqref{weights_recursion_x}-\eqref{theta_x_covariance_estimate} for all agents in the network. This he can do because the network topology and private signal models are common knowledge. The updates \eqref{linear_action} and \eqref{LMMSE_updates_mean_filter}-\eqref{LMMSE_updates_state_mean_filter} are performed for agent $i$'s own index only. 

The signal updates on \eqref{LMMSE_updates_mean_filter}-\eqref{LMMSE_updates_state_mean_filter} are illustrated inside the dashed box in Fig. \ref{feedback_diagram}. At time $t$, the inputs to the filter are the observed actions $\bba_{n(i)}(t)$ of agent $i$'s neighbors. The prediction $\bbE_{i,t}[\bba_{n(i)}(t)]=H_{i,t}\bbE_{i,t}[\bbx]$ of this vector is subtracted from the observed value and the resultant error is fed into two parallel blocks respectively tasked with updating the belief $\bbE_{i,t}[\theta]$ on the state of the world $\theta$, and the belief $\bbE_{i,t}[\bbx]$ on the private signals $\bbx$ of other agents. The error $\bba_{n(i),t}-\bbE_{i,t}[\bba_{n(i),t}]$ is multiplied by the gain $K^i_\bbx(t)$ and the resultant innovation is added to the previous mean estimate to correct the estimate of $\bbx$ [cf. \eqref{LMMSE_updates_mean_filter}]. Similarly, the error is multiplied by the gain $K^i_\theta(t)$ and the resultant innovation is added to the previous mean estimate to correct the estimate of $\theta$ at $i$ [cf. \eqref{LMMSE_updates_state_mean_filter}]. In order to determine the equilibrium play as per \eqref{linear_action}, agent $i$ multiples his private signal estimate $\bbE_{i,t}[\bbx]$ by the vector $\bbv_i(t)$ obtained by solving the system of linear equations in \eqref{BNE_matrix_form}.

Observe that in the QNG filter, we do not use the fact that estimates $\bbE_{i,t} \left[\theta\right]$ and $\bbE_{i,t}[\bbx]$ as well as actions $a_{i,t}$ can be written as linear combinations of the private signals [cf. \eqref{linear_estimates_first_one} and \eqref{linear_action_assumption}]. While the expressions in \eqref{linear_estimates_first_one} and \eqref{linear_action_assumption} are certainly correct, they cannot be used for implementation because $\bbx$ is only partially unknown to agent $i$. The role of \eqref{linear_estimates_first_one} and \eqref{linear_action_assumption} is to allow derivation of recursions that we use to keep track of the gains used in the QNG filter. These recursions can be divided into a group of LMMSE updates and a group of game updates as we show in Fig. \ref{gains_diagram}.

As it follows from \eqref{eqn_lmmse_gain_x}-\eqref{eqn_lmmse_gain_theta} and \eqref{LMMSE_updates_covariance}-\eqref{theta_x_covariance_estimate}, the update of LMMSE coefficients is identical to the gain and covariance updates of a sequential LMMSE. The only peculiarity is that the observation matrix $H_{j,t}$ is fed from the game update block and is partially determined by the LMMSE gains and covariances of previous iterations. Nevertheless, this peculiarity is more associated with the game block than with the LMMSE block. The game block uses \eqref{weights_recursion_x} and \eqref{weights_recursion_theta} to keep track of the matrices $L_{j,t}$ and the vectors $\bbk_{j,t}$. The matrices $L_{j,t}$ are used as building blocks of the matrix $L_t$ and the vectors $\bbk_{j,t}$ are stacked in the vector $\bbk_{t}$ and used to formulate the systems of equations in \eqref{BNE_matrix_form}. Solving this system of equations, using $L_t^{-1}$ when it is full rank or its pseudo inverse when it is not, yields the coefficients $\bbv_{j,t}$ which in turn determine the observation matrix $H_{j,t}$ as per \eqref{observation_matrix}. As mentioned before, the game block feeds the matrices $H_{j,t}$ to the filter block as they are used in the LMMSE gains and covariance updates. The LMMSE block feeds the gains $K^j_\bbx(t)$ and $K^j_\theta(t)$ to the game block as these are needed to update  $L_{j,t}$ and $\bbk_{j,t}$. 

We remark that agent $i$ is keeping track of the matrices and vectors in Fig. \ref{gains_diagram} for all $i \in V$. I.e., agent $i$ calculates observation matrices $H_{j,t}$ for $j \in V$ in the game block which are fed into the LMMSE block to obtain gains matrices $K^j_\bbx(t)$ and $K^j_\theta(t)$ for all $j \in V$. These gains are fed into the game block from the LMMSE block as they are needed to update  $L_{j,t}$ and $\bbk_{j,t}$ for all $j \in V$. The reason for this is the step in the game block in which we compute the play coefficients $\bbv_{j,t}$. To solve this system of equations, agent $i$ needs to build the matrix $L_t$ that is formed by the blocks $L_{j,t}$ of all agents. All of these computations for the coefficients of other agents are internal to agent $i$ and independent of the game realization. The gains can be computed offline prior to running the game. 

%In the calculation of game and LMMSE coefficients in Fig. \ref{gains_diagram}, we make use of the assumptions that signal and network structure, and the strategy profile are common knowledge. Assuming common knowledge of the signal structures guarantees that agent $i$ knows the initial estimation weights of all agents, that is, $\{ L_{j,0}\}_{j \in V}$ and $\{ \bbk_{j,0}\}_{j \in V}$. This information is used to calculate initial action weights $\{\bbv_{j,t}\}_{j \in V}$ by solving \eqref{linear_action} for $t= 0$ given common knowledge of the strategy profile. When the network structure and the strategy profile are common knowledge, agent $i$ can calculate the action coefficients $\{\bbv_{j,t}\}_{j \in V}$ for all $t$ and update her estimates locally. Further,  common knowledge of the network and Bayesian updates imply that agent $i$ can recursively update the estimation weights for $\bbx$ and $\theta$ for all agents. 

\begin{remark} \normalfont
The QNG filter can also be used in repeated games with purely informational externalities. In this case each agent's payoff is given by $u(\theta, a_i)  = -(\theta - a_i)^2$, and the problem is thus equivalent to the distributed estimation of the world state $\theta$ \cite{mossel2010efficient}. Our model subsumes the games with purely informational externalities as a special case. Given this payoff function, the best response of agent $i$ at time $t$ is the action $a_i(t) = \mathbf{E}_{i,t}[\theta]$. Hence, it is not necessary to solve \eqref{BNE_matrix_form} for the optimal strategy coefficients $\bbv_{i,t}$. Other than this the QNG filter remains unchanged. Since in the case of purely informational externalities the end goal is the estimation of $\theta$, the QNG filter is tantamount to an optimal distributed implementation of a Kalman filter.
\end{remark}

\begin{figure*} 
\centering
%!TEX root = bqng_arxiv.tex

\tikzstyle{name box}= [draw, text centered, anchor=north east,
  						       minimum height = 1.2cm, minimum width = 1.3cm, text width=1.1cm, fill=blue!10]
\tikzstyle{shorter name box}= [draw, text centered, anchor=north east,
  						       minimum height = 1.2cm, minimum width = 1.1cm, text width=0.9cm, fill=red!10]

\tikzstyle{equation box} = [draw, text justified, minimum height = 1.2cm, minimum width = 6.6cm, anchor = west, text width=6.2cm,fill=blue!10]

\tikzstyle{shorter equation box} = [draw, text justified, minimum height = 1.2cm, minimum width = 6.2cm, anchor = west, text width=5.8cm, fill=red!10]

{\fontsize{8}{8}\selectfont \begin{tikzpicture}[scale=0.9]

  % GAME variables
  \node                    [shorter name box] (G1) {Variable};     
  \node at (G1.south east) [shorter name box] (G2) {$L_{j,t}$};   
  \node at (G2.south east) [shorter name box] (G3) {$\bbk_{j,t}$};
  \node at (G3.south east) [shorter name box, fill=red!20] (G4) {$\bbv_{j,t}$};
  \node at (G4.south east) [shorter name box] (G5) {$H_{j,t}$};
  \path    (G1.north east) ++ (2.5,0.5) node {Game coefficients};

  % GAME updates
  \node at (G1.east) [shorter equation box] (G1) {Update};
  \node at (G2.east) [shorter equation box] (G2)
      {$L_{j,t+1} = L_{j,t} +  K^j_\bbx(t) \Big(H_{j,t}^T- H_{j,t}^TL_{j,t} \Big)$ \hfill
       \eqref{weights_recursion_x}};   
  \node at (G3.east) [shorter equation box] (G3)
      {$\bbk_{j,t+1}^T = \bbk_{j,t}^T + K^j_\theta(t) \Big(H_{j,t}^T - H_{j,t}^T L_{j,t}\Big)$ \hfill
       \eqref{weights_recursion_theta}};
  \node at (G4.east) [shorter equation box, thick, fill=red!20] (G4)
      {$ L_t  \bbv_t = \delta \bbk_t$\hfill
       \eqref{BNE_matrix_form}};
  \node at (G5.east) [shorter equation box] (G5)
      {$H_{j,t}:= \hspace{-1mm}\big[\bbv_{k_{j,1},t}^T L_{k_{j,1},t}; \ldots; \bbv_{k_{j,d(j)},t}^TL_{k_{j,d(j)},t} ]^T$ \hfill
       \eqref{observation_matrix}};

  % FILTER variables
  \node at (10,0)          [name box] (F1) {Variable};     
  \node at (F1.south east) [name box] (F2) {$K^j_\bbx(t)$};   
  \node at (F2.south east) [name box] (F3) {$K^j_\theta(t)$};
  \node at (F3.south east) [name box] (F4) {$M^j_{\bbx \bbx} (t)$};
  \node at (F4.south east) [name box] (F5) {$M^j_{\theta \bbx} (t)$};   
  \path    (F1.north east) ++ (2.5,0.5) node {LMMSE coefficients};

  % FILTER updates
  \node at (F1.east) [equation box] {Update};
  \node at (F2.east) [equation box] 
      {$K^j_\bbx(t) = M^j_{\bbx\bbx}(t)H_{j,t}\Big(H_{j,t}^T M^j_{\bbx\bbx}(t)H_{j,t}\Big)^{-1}$ \hfill
       \eqref{eqn_lmmse_gain_x}};   
  \node at (F3.east) [equation box] 
      {$K^j_\theta(t) = M^j_{\theta\bbx}(t)H_{j,t}
         \Big(H_{j,t}^TM^j_{\bbx \bbx}(t)H_{j,t}\Big)^{-1}$ \hfill
         \eqref{eqn_lmmse_gain_theta}};
  \node at (F4.east) [equation box]
      {$M^j_{\bbx \bbx} (t+1) 
         = M^j_{\bbx \bbx} (t)- K^j_\bbx(t) H_{j,t}^T M^j_{\bbx \bbx} (t)$\hfill
       \eqref{LMMSE_updates_covariance}};
  \node at (F5.east) [equation box] (FF5)
      {$M^j_{\theta \bbx} (t+1) 
         \ =\ M^j_{\theta \bbx} (t) -  K^j_\theta(t) H_{j,t}^T M^j_{\bbx \bbx} (t)$ \hfill
       \eqref{theta_x_covariance_estimate}};

   % ARROWS BETWEEN BLOCKS
   \path[red,-stealth,shorten >= 2pt, shorten <= 2pt, thick]  
        (G5.east) edge [above] node {$H_{j,t}$} (F5.west) ; 
   \path[blue,-stealth,shorten >= 2pt, shorten <= 2pt, thick]  
        (F2.west) edge [above] node {$K^j_\bbx(t)$} (G2.east) ; 
   \path[blue,-stealth,shorten >= 2pt, shorten <= 2pt, thick]  
        (F3.west) edge [above] node {$K^j_\theta(t)$} (G3.east) ; 

   % ARROWS OUT OF GAME TOWARDS FILTER
   \path[red, draw, -stealth,shorten >= 2pt, shorten <= 2pt, thick]  
        (G5.south) ++ (-1.1,0) -- ++ (0,-1.0)node [left] {$\bbv_{i,t}$} -- ++ (0,-0.5);     
   \path[red, draw, -stealth,shorten >= 2pt, shorten <= 2pt, thick]  
        (G5.south) ++ (+1.1,0) -- ++ (0,-1.0)node [left] {$H_{i,t}$} -- ++ (0,-0.5);
   \path(G5.south) ++ (+1.1,0) ++ (0,-1.8) node {to QNG  filter};
   \path(G5.south) ++ (-1.1,0) ++ (0,-1.8) node {to QNG  filter};

   % ARROWS OUT OF FILTER TOWARDS FILTER
   \path[blue, draw, -stealth,shorten >= 2pt, shorten <= 2pt, thick]  
        (FF5.south) ++ (-1.1,0) -- ++ (0,-1.0)node [left] {$K^i_\bbx(t)$} -- ++ (0,-0.5);
   \path[blue, draw, -stealth,shorten >= 2pt, shorten <= 2pt, thick]  
        (FF5.south) ++ (+1.1,0) -- ++ (0,-1.0)node [left] {$K^i_\theta(t)$} -- ++ (0,-0.5);
   \path(FF5.south) ++ (+1.1,0) ++ (0,-1.8) node {to QNG filter};
   \path(FF5.south) ++ (-1.1,0) ++ (0,-1.8) node {to QNG  filter};

\end{tikzpicture}}\vspace{-3mm}
\caption{Propagation of gains required to implement the Quadratic Network Game (QNG) filter of Fig. \ref{feedback_diagram}. Gains are separated into interacting LMMSE and game blocks. All agents perform a full network simulation in which they compute the gains of all other agents. This is necessary because when we compute the play coefficients $\bbv_{j,t}$ in the game block, agent $i$ builds the matrix $L_t$ that is formed by the blocks $L_{j,t}$ of all  agents [cf. \eqref{multipliers_matrix}]. This full network simulation is possible because the network topology and private signal models are common knowledge.}\vspace{-5mm}
\label{gains_diagram}
\end{figure*}
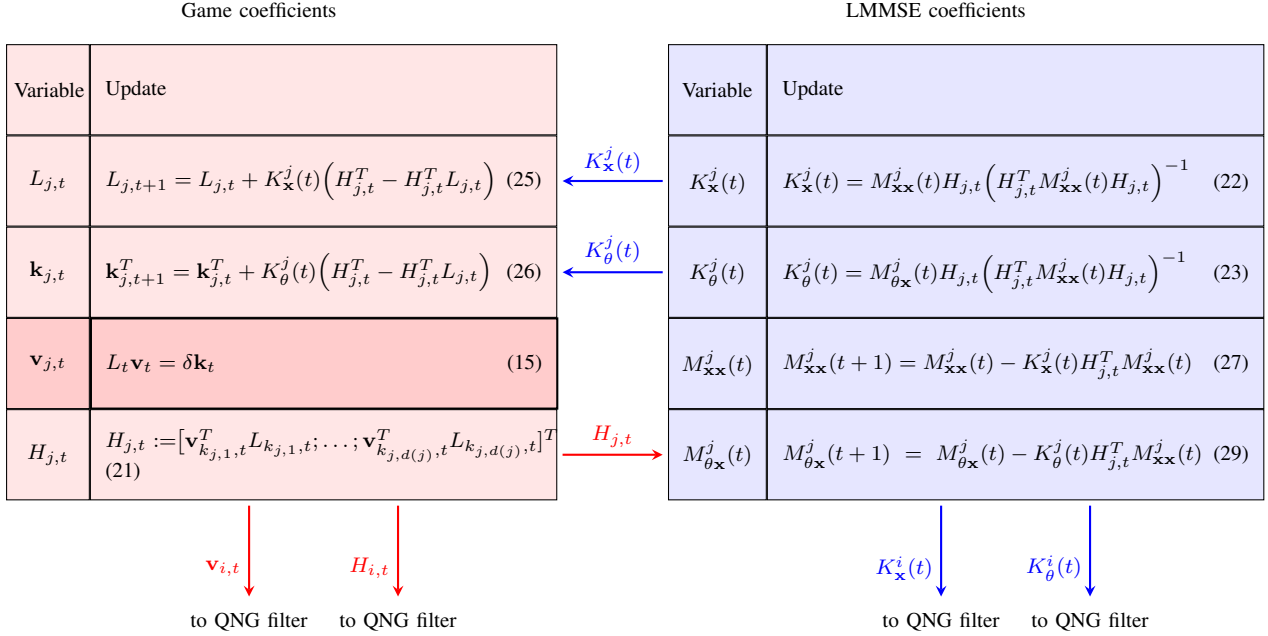

%!TEX root = bqng_arxiv.tex

\section{Vector states and vector observations} \label{vector_state} 

%\blue{Need to rewrite in the image of the rewritten Section \ref{recursion}.}
Consider the case when state of the world is a vector, that is,  $\bbtheta \in \reals^m$ for $m > 1$. Similar to the scalar case, each agent receives initial private signal $\bbx_i \in \reals^m$,
\begin{align} \label{private_signal_vector}
\bbx_i = \bbtheta + \bbepsilon_i
\end{align}
where the additive noise term $\bbepsilon_i \in \reals^{m}$ is multivariate Gaussian with zero mean and variance-covariance matrix $C_{i} \in \reals^{m \times m}$. For future reference, define the vector obtained by stacking elements at the $k$th row and $l$th column of variance-covariance matrices of all agents, $\bbC_{k,l} := [C_1[k,l], \ldots, C_N[k,l]]^T$. We use $\bbx_i[n]$ to denote the $n$th private signal of agent $i$ where $n \leq m$. We assume that private signals are independent among agents, that is,  $\bbE_{i,0}[\bbepsilon_i \bbepsilon_j] = 0$ for all $i \in V$ and $j \in V \setminus \{i\}$. We define the set of all private signals as 
\begin{equation} \label{private_signals_vector}
\bbx := [\bbx_1[1], \ldots, \bbx_N[1],  \ldots, \bbx_1[m], \ldots, \bbx_N[m]]^T,
\end{equation}
where $\bbx \in \reals^{N m \times 1}$. We use $\bbx[n] := [ \bbx_1[n], \dots, \bbx_N [n]]^T$ to denote the vector of private signals of agents on the $n$th state of the world. 

At each stage $t$, agent $i$ takes action $\bba_i(t) \in \reals^m$.
Agent $i$'s action at time $t$ is to maximize a payoff function which is represented by the following quadratic function
\begin{equation} \label{quadratic_utility_form_vector}
u_i(\bba_i, \{a_{j}\}_{j \in V\setminus i} , \bbtheta)  =  - \frac{1}{2} \sum_{j \in V}  \bba_j^T \bba_j + \hspace{-3mm}\sum_{j \in V\setminus\{i\}} \hspace{-3mm} \bba_i^T B_{ij}  \bba_j  +\bba_i^T D \bbtheta,
\end{equation}
where constants $B_{ij}$ and $D$ belong to $\reals^{m \times m}$. Similar to the scalar case, other additive terms that depend on $\{\bba_{j}\}_{j \in V \setminus i}$ and $\bbtheta$ can exist without changing the results to follow. We obtain the best response function for agent $i$ by taking the derivative of the expected utility function with respect to $\bba_i$, equating it to zero, and solving for $\bba_i$:
\begin{equation}\label{linear_reply_vector}
\BR_{i,t}(\{\sigma_{j,t}(h_{j,t})\}_{j \in V \setminus i}) = \sum_{j\in V\setminus{i}}\hspace{-0.2cm} B_{ij} \bbE_{i,t}[\sigma_{j,t}(h_{j,t})] + D \bbE_{i,t}[\bbtheta].
\end{equation}
Note that $\BR_{t}: \reals^{Nm} \to \reals^{Nm}$.

Similar to the case when the unknown parameter is a scalar, it is sufficient for agents to keep track of estimates of $\bbx$ in order to achieve the best estimate of $\bbtheta$. Accordingly, the definitions of estimates of private signals and the unknown parameters and their corresponding covariance matrices \eqref{error_covariance_xx}--\eqref{error_covariance_thetax} are the same as in the scalar case. 

In what follows, we show that the mean estimates are linear in private signals and equilibrium actions are linear in expectations of private signals in the similar fashion we did for the scalar state of the world. 
\begin{lemma} \label{linear_equilibrium_at_time_t_vector}
Consider a Bayesian game with quadratic utility as in \eqref{quadratic_utility_form_vector}. Suppose that for all agents $i$, the joint posterior beliefs on the state of the world $\bbtheta$ and the private signals $\bbx$ given the local history $h_{i,t}$ at time $t$, $\bbP_{i,t}([\bbtheta^T, \bbx^T])$, are Gaussian with means expressed as
\begin{align} \label{linear_estimates_first_one_vector}
\mathbf{E}_{i,t} \left[\bbtheta\right] = Q_{i,t}  \bbx ,  \text{ and } \bbE_{i,t}[\bbx] = L_{i,t} \bbx,
\end{align}
where $L_{i,t} \in \reals^{Nm \times Nm}$ and $Q_{i,t} \in \reals^{m \times Nm}$ are known estimation weights.
If there exists an equilibrium strategy profile that is linear in expectations of private signals,
\begin{equation} \label{linear_action_vector}
 \mathbf{\sigma}_{i,t}^*(h_{i,t}) = U_{i,t} \bbE_{i,t}[\bbx]   \quad\forall i \in V,
\end{equation}
then the action coefficients $\{ U_{i,t}\}_{i\in V}$ can be obtained by solving the system of linear equations
\begin{equation} \label{linear_equilibrium_vector_thm}
L_{i,t}^T U_{i,t}^T   =  \sum_{j\in V\setminus{i}}  L_{i,t}^T L_{j,t}^T U_{j,t}^T B_{ij}^T  +  Q_{i,t}^T D^T, \quad \forall i \in V
\end{equation}
\end{lemma}
\begin{proof}
The proof is analogous to the proof of Lemma \ref{linear_equilibrium_at_time_t}. By substituting the candidate strategies in \eqref{linear_action_vector} to the best response function in \eqref{linear_reply_vector} for all $i \in V$, we obtain the following equilibrium equations
\begin{equation} \label{equilibrium_equations}
U_{i,t} \bbE_{i,t} [\bbx] = \sum_{j \in V \setminus \{i\}} B_{ij} \bbE_{i,t}[U_{j,t} \bbE_{j,t} [\bbx]] + D \bbE_{i,t}[\bbtheta].
\end{equation}
 for all $i \in V$. After using the fact that $\bbE_{i,t}[\bbE_{j,t}[\bbx]] = L_{j,t} \bbE_{i,t}[\bbx]$ with mean estimate assumptions in  \eqref{linear_estimates_first_one_vector} for the corresponding terms in \eqref{equilibrium_equations}, we obtain the following set of equations
\begin{equation} \label{equilibrium_equations_2}
U_{i,t} L_{i,t} \bbx = \sum_{j \in V \setminus \{i\}} B_{ij} U_{j,t} L_{j,t} L_{i,t}\bbx + D Q_{i,t} \bbx.
\end{equation}
We ensure that the strategies in \eqref{linear_action_vector} satisfy the equilibrium equations for any realization of history by equating coefficients that multiply each component of $\bbx$ in \eqref{equilibrium_equations_2} which yields the set of equations given by \eqref{linear_equilibrium_vector_thm}. 
\end{proof}

For a linear equilibrium strategy, the actions can be written as a linear combination of the private signals using \eqref{linear_estimates_first_one_vector}, that is,  the action of agent $i$ at time $t$ is given by
\begin{equation} \label{linear_action_assumption_vector}
a_{i}(t) = U_{i,t} L_{i,t} \bbx \qquad \forall i \in V. 
\end{equation}
Being able to express actions as in \eqref{linear_action_assumption_vector} permits writing observations of agents in linear form. From the perspective of an observer, the action $\bba_{j}(t)$ is equivalent to observing a linear combination of private signals. As a result, we can represent observation vector of agent $i$ $\bba_{n(i)}(t):=\big[\bba_{j_1}(t) , \ldots, \bba_{j_{d(i)}}(t)\big]^T \in \reals^{m d(i) } $ in linear form as
\begin{align} \label{observation_matrix_BNE_vector}
\bba_{n(i)}(t) = H_{i,t}^T \bbx = [U_{j_1,t} L_{j_1,t}; \ldots; U_{j_{d(i)},t} L_{j_{d(i)},t}]\bbx
\end{align} 
where $H_{i,t}^T = [U_{j_1,t} L_{j_1,t}; \ldots; U_{j_{d(i)},t} L_{j_{d(i)},t}]\in \reals^{m d(i) \times N m}$ is the observation matrix of agent $i$. 

Agent $i$'s belief of $\bbx$ at time $t$ is normal, and at time $t+1$ agent $i$ observes a linear combination of $\bbx$. Hence, agent $i$'s belief at time $t+1$ can be obtained by a sequential LMMSE update. As a result, mean estimates remain weighted sums of private signals as in \eqref{linear_estimates_first_one_vector}. In the following lemma, we explicitly present the way we compute the estimation weights, $L_{i,t+1}$ and $Q_{i,t+1}$, at time $t+1$ when $\bbtheta \in \reals^m$. 

%%%%%%%%%%%%%%%%%%%%%%%
%%%%%%%%%%%%%% Theorem 
%%%%%%%%%%%%%%%%%%%%%%%
\begin{lemma} \label{rational_belief_updates_theorem_L_vector}
Consider a Bayesian game with quadratic function as in \eqref{quadratic_utility_form_vector} and the same assumptions and definitions of Lemma \ref{linear_equilibrium_at_time_t_vector}.  Further define the gain matrices  as
\begin{IEEEeqnarray}{rCl}
K^i_\bbx(t) &:=&  M^i_{\bbx \bbx} (t) H_{i,t}\big(H_{i,t}^T  M^i_{\bbx \bbx} (t) H_{i,t}\big)^{-1}, \label{eqn_lmmse_gain_x_vector} \\
K^i_{\bbtheta} (t) &:=& M^i_{\bbtheta \bbx} (t) H_{i,t} \big(H_{i,t}^T  M^i_{\bbx \bbx} (t) H_{i,t}\big)^{-1}.\label{eqn_lmmse_gain_theta_vector}
\end{IEEEeqnarray}
If agents play according to a linear equilibrium strategy then agent $i$'s posterior $\bbP_{i, t+1}([\bbtheta^T, \bbx^T])$ is Gaussian with means that are linear combination of private signals,
\begin{align} \label{linear_estimates_second_one_vector}
\mathbf{E}_{i,t+1} \left[\bbtheta\right] = Q_{i,t+1}  \bbx ,  \text{ and } \bbE_{i,t+1}[\bbx] = L_{i,t+1} \bbx,
\end{align}
where the estimation matrices are given by
\begin{align}
L_{i,t+1} &= L_{i,t} +  K^i_\bbx(t) \left(H_{i,t}^T -   H_{i,t}^T L_{i,t}\right) \label{weights_recursion_x_vector},\\
Q_{i,t+1} &= Q_{i,t} + K^i_{\bbtheta}(t) \left( H_{i,t}^T -  H_{i,t}^T, L_{i,t}\right)\label{weights_recursion_theta_vector},
\end{align} 
and the covariance matrices are further given by
\begin{align}
M^i_{\bbx \bbx} (t+1) =& M^i_{\bbx \bbx} (t) - K^i_\bbx(t) H_{i,t}^T M^i_{\bbx \bbx} (t),\label{LMMSE_updates_covariance_vector} \\
M^i_{\bbtheta \bbtheta} (t+1) =& M^i_{\bbtheta \bbtheta} (t) - \left[K_{\bbtheta}^i(t)^T H_{i,t}^T M^i_{ \bbx \bbtheta} (t)\right]^T , \label{theta_covariance_estimate_vector} \\
M^i_{\bbtheta \bbx} (t+1) =& M^i_{\theta \bbx} (t) -  K^i_{\bbtheta}(t) H_{i,t}^T M^i_{\bbx \bbx} (t).\label{theta_x_covariance_estimate_vector}
\end{align}
\end{lemma}
\begin{proof}
The proof is identical to the proof of Lemma \ref{rational_belief_updates_theorem_L} with the action coefficients $U_{i,t}$ taking the place of $\bbv_{i,t}$.
\end{proof}

Lemma \ref{rational_belief_updates_theorem_L_vector} shows that when mean estimates are linear combinations  of private signals at time $t$, they remain that way at time $t+1$. In the next theorem, we show that assumption in \eqref{linear_estimates_first_one_vector} is indeed true for all time by realizing that the estimates at time $t =0$ are linear combinations of private signals. To simplify presentation of initial conditions, we assume that agent $i$'s private signals are independent, $\bbE_{i,0}[\bbx_i[k] \bbx_i[l]] = 0$ for all $k =1, \dots, m$ and $l \neq k$.  
 
%The result relies on the fact that initial mean estimates of $\bbx$ and $\theta$ are linear in private signals and as a result there exists an initial strategy that is linear in mean estimates of private signals satisfying equations in \eqref{Bayes-Nash}. Since initial mean estimates of private signals, $\bbE_{i,0}[\bbx]$ are linear combinations of privates signals, initial actions $a_{n(i)}(0)$ are linear combinations of $\bbx$ as well. Consequently, we can apply the estimate updates in Theorem \ref{rational_belief_updates_theorem_L} to obtain estimates at time $t=1$ that are linear combinations of private signals. Inductively, the same process continues where given mean estimates that are linear combinations of private signals, there exists a strategy linear in estimates of private signals that satisfies \eqref{Bayes-Nash}. Hence, mean estimates of $\bbx$ and $\theta$ continue to be linear functions of $\bbx$.

%
%%%%%%%%%%%%%%%%%%%%%%%
%%%%%%%%%%%%%% Theorem 
%%%%%%%%%%%%%%%%%%%%%%%

\begin{theorem} \label{rational_updates_theorem_vector}
Given the quadratic utility function in \eqref{quadratic_utility_form_vector}, if there exists a linear equilibrium strategy $\sigma_t^*$ as in \eqref{linear_action_vector} for $t \in \naturals$, then the action coefficients $U_{i,t}$ can be computed by solving the system of linear equations in \eqref{linear_equilibrium_vector_thm}, and further, agents' estimates of $\bbx$ and $\bbtheta$ are linear combinations of private signals as in \eqref{linear_estimates_first_one_vector} with estimation matrices computed recursively using \eqref{eqn_lmmse_gain_x_vector}-\eqref{eqn_lmmse_gain_theta_vector} and \eqref{weights_recursion_x_vector}-\eqref{theta_x_covariance_estimate_vector} with initial values
\begin{equation}  \label{state_initial_estimate}
Q_{i,0}:= %\kappa_{i}
\begin{psmallmatrix}
 \bbe_i^T & \bbzero_{1 \times N}& \ldots &  \bbzero_{1 \times N} \\
  \bbzero_{1 \times N} &  \bbe_i^T & \ldots & \bbzero_{1 \times N} \\
   \vdots & \cdots & \ddots &\vdots \\
    \bbzero_{1 \times N} & \ldots &  \bbzero_{1 \times N} &  \bbe_i^T 
 \end{psmallmatrix} 
 \in \reals^{m \times Nm},
 \end{equation}
\begin{equation} \label{signal_initial_estimate}
L_{i,0}:= 
{\rm diag}\left(\left[\bbone \bbe_i^T, \ldots,\bbone \bbe_i^T \right]\right)
%\begin{psmallmatrix} \label{signal_initial_estimate}
%\bbone \bbe_i^T & \bbzero_{N \times N}& \ldots &  \bbzero_{N \times N} \\
%  \bbzero_{N \times N} &  \bbone \bbe_i^T & \ldots & \bbzero_{N \times N} \\
%   \vdots & \cdots & \ddots &\vdots \\
%    \bbzero_{N \times N} & \ldots &  \bbzero_{N \times N} &  \bbone \bbe_i^T 
% \end{psmallmatrix}
  \in \reals^{Nm \times Nm},
 \end{equation}
where $\bbe_i \in \reals^{N}$. The initial covariance matrix $M^i_{\bbx \bbx}(0) \in \reals^{N m \times N m}$ is a  diagonal block matrix with $N \times N$ blocks $((M^i_{\bbx \bbx}))_{k,k} \in \reals^{N \times N}$ for $k = 1, \dots, m$ , initial variance $M^i_{\bbtheta \bbtheta}(0)  \in \reals^{m \times m}$ and initial cross covariance $M^i_{\bbtheta \bbx}(0) \in \reals^{m \times N m}$ are given by
\begin{align}
\left((M^i_{\bbx \bbx})\right)_{k,k} &= {\rm diag}(\bar \bbe_i) {\rm diag}(\bbC_{k,k}) + \bar\bbe_i \bar\bbe_i^T C_i[k,k],  \label{variance_vector}\\
M^i_{\bbtheta \bbtheta}(0) &= C_{i}, \label{state_variance} \\
 M^i_{\bbtheta \bbx}(0) &=  C_i 
 \begin{psmallmatrix} 
 \bar\bbe_i^T & \bbzero_{1 \times N}& \ldots &  \bbzero_{1 \times N} \\
  \bbzero_{1 \times N} &  \bar\bbe_i^T & \ldots & \bbzero_{1 \times N} \\
   \vdots & \cdots & \ddots &\vdots \\
    \bbzero_{1 \times N} & \ldots &  \bbzero_{1 \times N} &  \bar\bbe_i^T 
 \end{psmallmatrix} \label{state_private_variance_matrix}
\end{align}
%\end{align}
%%\red{
%%$M^i_{\bbx_j \bbx_j}(0) = C_{j} + C_{i}$ for all $j \in V \setminus \{i\}$, 
%%$M^i_{\bbx_i \bbx_k}(0) = 0$ for all $k \in V$, and
%%$M^i_{\bbx_j \bbx_k}(0) = C_{i}$ for all $k \in V \setminus \{i\}$ and $j \in V \setminus \{i,k\}$,}
%%
%\begin{align}
%
%where matrix $\kappa_{i,0} \in \reals^{m \times Nm}$ is defined as
%\begin{align}
%\kappa_{i,0} := \begin{psmallmatrix} \label{state_private_variance_matrix}
% \bar\bbe_i^T & \bbzero_{1 \times N}& \ldots &  \bbzero_{1 \times N} \\
%  \bbzero_{1 \times N} &  \bar\bbe_i^T & \ldots & \bbzero_{1 \times N} \\
%   \vdots & \cdots & \ddots &\vdots \\
%    \bbzero_{1 \times N} & \ldots &  \bbzero_{1 \times N} &  \bar\bbe_i^T 
% \end{psmallmatrix}
% \end{align}
\end{theorem}
%where the matrix IdNd is the horizontal concatenation of
%N instances of d  d identity matrices, i.e., IdNd =
%[Idd; : : : ; Idd]. The
%
\begin{proof} See Appendix \ref{Theorem_2_proof}.
\end{proof}

Similar to the scalar case, when network structure and the equilibrium strategy profile are common knowledge, agent $i$ can calculate the weights $\{U_{j,t}\}_{j \in V}$ for all $t$ and update his estimates locally. In Algorithm \ref{alg2}, we provide a sequential local algorithm for agent $i$ to calculate updates for $\bbtheta$ and $\bbx$ and to act according to equilibrium strategy. The Bayesian rational learning defined here in Algorithm \ref{alg2} for the vector state case follows the same steps for the scalar case defined in Section \ref{QNG_filter} and by Figs. \ref{feedback_diagram} and \ref{gains_diagram}.

\begin{algorithm}[t]
\caption{QNG filter for $\bbtheta \in \reals^d$}          % give the algorithm a caption
\label{alg2}                           % and a label for \ref{} commands later in the document
\emph{Initialization:} Set posterior distribution on $\bbtheta$ and $\bbx$
\begin{equation}
 \begin{bmatrix} \bbtheta \\ \bbx \end{bmatrix}\given h_{i,0} \sim \ccalN\left(\begin{bmatrix} Q_{i,0} \bbx\\ L_{i,0} \bbx \end{bmatrix}, 
\begin{pmatrix}
M_{\bbtheta \bbtheta}^i(0) ,  M_{\bbtheta \bbx}^i(0) \\
M_{\bbx \theta}^i(0) , M_{\bbx \bbx}^i(0)
\end{pmatrix}
 \right) \nonumber
\end{equation}
and 
$\{L_{j,0}, \bbk_{j,0}\}_{j\in V}$ according to \eqref{state_initial_estimate} and  \eqref{signal_initial_estimate}.   \vspace{4mm}\\
\vspace{2mm}
{\bf For} $t = 0, 1, 2,\ldots$  \vspace{-2mm}\\
\begin{enumerate}
\item \emph{Equilibrium strategy:}  Solve for $\{U_{j,t}\}_{j \in V}$ using the set of equations in \eqref{linear_equilibrium_vector_thm}. \\ \vspace{-2mm}
\item \emph{ Play and observe:} Take action $a_{i}(t) = U_{i,t} \bbE_{i,t}[\bbx]$ and observe $a_{n(i)}(t)$. \\ \vspace{-2mm}
\item \emph{Observation matrix:} Construct $H_{i,t}$ using \eqref{observation_matrix_BNE_vector}. \\ \vspace{-2mm}
%\begin{align} \label{observation_matrix_BNE_vector}
%H_{i,t}^T = [U_{j_1,t} L_{j_1,t}, \ldots, U_{j_{d(i)},t} L_{j_{d(i)},t}].
%\end{align} 
%where $U_{j_{k},t}$ is the weight vector in \eqref{linear_action_vector} and $L_{j_{k},t}$ is the estimation weight corresponding to the $k$th neighbor of $i$. \\ \vspace{-2mm}
\item \emph{Bayesian estimates:}  Update $\bbE_{i,t}[\bbx]$ and $\bbE_{i,t}[\theta]$ using \eqref{LMMSE_updates_mean} and \eqref{theta_mean_estimate}, respectively. Update error covariance matrices using \eqref{LMMSE_updates_covariance_vector}--\eqref{theta_x_covariance_estimate_vector}.\\ \vspace{-2mm}
\item \emph{Estimation weights:} Update $\{L_{j,t}, \bbk_{j,t}\}_{j\in V}$ using \eqref{weights_recursion_x_vector}--\eqref{weights_recursion_theta_vector}.
\end{enumerate}
\end{algorithm}
\vspace{-3mm}
%Note that the sequential LMMSE estimator updates in equations \eqref{LMMSE_updates_mean} and \eqref{theta_mean_estimate} can be used for local estimate updates of private signals in \eqref{private_signals_vector} and $\bbtheta$ with the observation matrix defined as
%%
%\begin{align} \label{observation_matrix_BNE_vector}
%H_{i,t}^T = [U_{j_1,t} L_{j_1,t}, \ldots, U_{j_{d(i)},t}^T L_{j_{d(i)},t}]
%\end{align}
%%
%where $U_{j_k,t}$ is the weight vector in \eqref{linear_action_vector} and $L_{j_k,t}$ is the estimation weight corresponding to the $k$th neighbor of $i$. 
%!TEX root = bqng_arxiv.tex

\section{Cournot Competition}\label{cournot}

In a Cournot competition model $N$ firms produce a common good that they sell in a market with limitless demand. The cost per production unit $c$ is common for all firms and constant for all times. The selling unit price, however, decreases as the total amount of goods produced by all companies increases. We adopt the specific linear model $p-\sum_{j \in V} a_j$ for the selling unit price, where $p$ is the constant market price when no goods are produced. The profit of firm $i$ for production level $a_i \in \reals^+$ is therefore given by the utility
\begin{equation}\label{cournot_profit}
   u_i(a_i, \{a_{j}\}_{j \in V\setminus i}, \theta) = - c a_i + (p -  a_i - \sum_{j \in V \setminus i} a_j)a_i.
\end{equation}
The utility function in \eqref{cournot_profit} is not of the quadratic form given in \eqref{quadratic_utility_form} because there are two information externalities, the cost $c$ and the clearing price $p$. While it is possible to resort to the vector form of the QNG filter covered in Section \ref{vector_state}, it is simpler to write \eqref{cournot_profit} in a form compatible with \eqref{quadratic_utility_form} by defining the parameter $\theta:= p-c$ as the effective unit profit at the market price. Using this definition in \eqref{cournot_profit} and reordering terms yields
\begin{equation} \label{cournot_profit_structured}
u_i(a_i, \{a_{j}\}_{j \in V\setminus i}, \theta) =  (\theta -  a_i - \sum_{j \in V \setminus i} a_j)a_i.
\end{equation}
Since this utility function is of the form in \eqref{quadratic_utility_form}, we can use the QNG filter of Section \ref{QNG_filter} as summarized in Figs. \ref{feedback_diagram} and \ref{gains_diagram} to determine subsequent BNE production levels. The explicit form of the equilibrium equation in \eqref{linear_reply} is
\begin{equation}\label{BR_cournot}
 \sigma^*_{i,t}(h_{i,t}) = \frac{1}{2}\bbE_{i,t}[\theta] -  \frac{1}{2}\sum_{j \in V \setminus i} \bbE_{i,t}[\sigma^*_{j,t}(h_{j,t})].
\end{equation}
It is immediate from \eqref{BR_cournot} that when $\bbE_{i,t}[\theta] < 0$ it is best for firm $i$ to shut down production. To avoid boundary conditions we restrict attention to cases where private signals $\bbx$ are such that $ \bbE_{i,t}[\theta] > 0$ for all $i \in V$ and $t \in\naturals$. This can be guaranteed if all private signals are nonnegative, i.e., $\bbx\geq\bbzero$.
In a game with complete information all private signals $\bbx$ are known to all agents. In this case the (regular) Nash equilibrium actions of all agents coincide and are given by
\begin{equation} \label{Cournot_Nash}
   a^*_i = \frac{\bbE[\theta \given \bbx]}{N+1} \quad \forall i \in V.
\end{equation}
The numerical simulations in the next section show that the BNE strategies in \eqref{BR_cournot} converge to the (regular) Nash equilibrium strategy \eqref{Cournot_Nash} in a finite number of steps.

\begin{figure}
\centering
\tikzstyle{agent} = [circle, draw=black, inner sep=0pt, minimum size=0.35cm, fill=blue!30]
\tikzstyle{arrow} = [stealth-stealth, thin]
{\tiny\begin{tikzpicture}[scale=0.8]
   % LINE NETWORK
   \path    (0,0) node (1) [agent]  {1} ++(1,0) node (2) [agent]  {2}
          ++(1,0) node (3) [agent]  {3} ++(1,0) node (4) [agent]  {4}
          ++(1,0) node (5) [agent]  {5};
   \draw [arrow] (1) to (2);   \draw [arrow] (2) to (3);
   \draw [arrow] (3) to (4);   \draw [arrow] (4) to (5);   
   % STAR NETWORK
   \path    ( 6, 0) node (5) [agent]  {5} +( 1, 1) node (1) [agent]  {1}
           +(-1,   1) node (2) [agent]  {2} +(-1,-1) node (3) [agent]  {3}
           +( 1,  -1) node (4) [agent]  {4};
   \draw [arrow] (5) to (1);   \draw [arrow] (5) to (2);
   \draw [arrow] (5) to (3);   \draw [arrow] (5) to (4);   
   % RING NETWORK
   \path    (9.5, 0) + (  0:1.5) node (1)  [agent]  {1}
                    + ( 36:1.5) node (2)  [agent]  {2}
                    + ( 73:1.5) node (3)  [agent]  {3}
                    + (108:1.5) node (4)  [agent]  {4}
                    + (144:1.5) node (5)  [agent]  {5}
                    + (180:1.5) node (6)  [agent]  {6}
                    + (216:1.5) node (7)  [agent]  {7}
                    + (252:1.5) node (8)  [agent]  {8}
                    + (288:1.5) node (9)  [agent]  {9}
                    + (324:1.5) node (10) [agent]  {10};              
   \draw [arrow, bend right=20] (1) to (2);   \draw [arrow, bend right=20] (2)  to (3);
   \draw [arrow, bend right=20] (3) to (4);   \draw [arrow, bend right=20] (4)  to (5);   
   \draw [arrow, bend right=20] (5) to (6);   \draw [arrow, bend right=20] (6)  to (7);
   \draw [arrow, bend right=20] (7) to (8);   \draw [arrow, bend right=20] (8)  to (9);
   \draw [arrow, bend right=20] (9) to (10);  \draw [arrow, bend right=20] (10) to (1);      
\end{tikzpicture}} \vspace{-1mm}
\caption{Line, star and ring networks.}
\label{line_star_ring}
\end{figure}
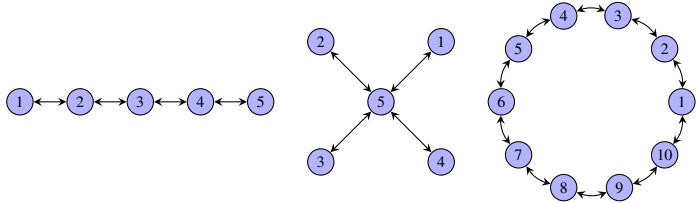

%%%%%%%%%%%%%%%%%%
\begin{figure}[!t]
\centering
\begin{tabular}{ccc}\hspace{-5mm}
\includegraphics[width=0.32\linewidth]{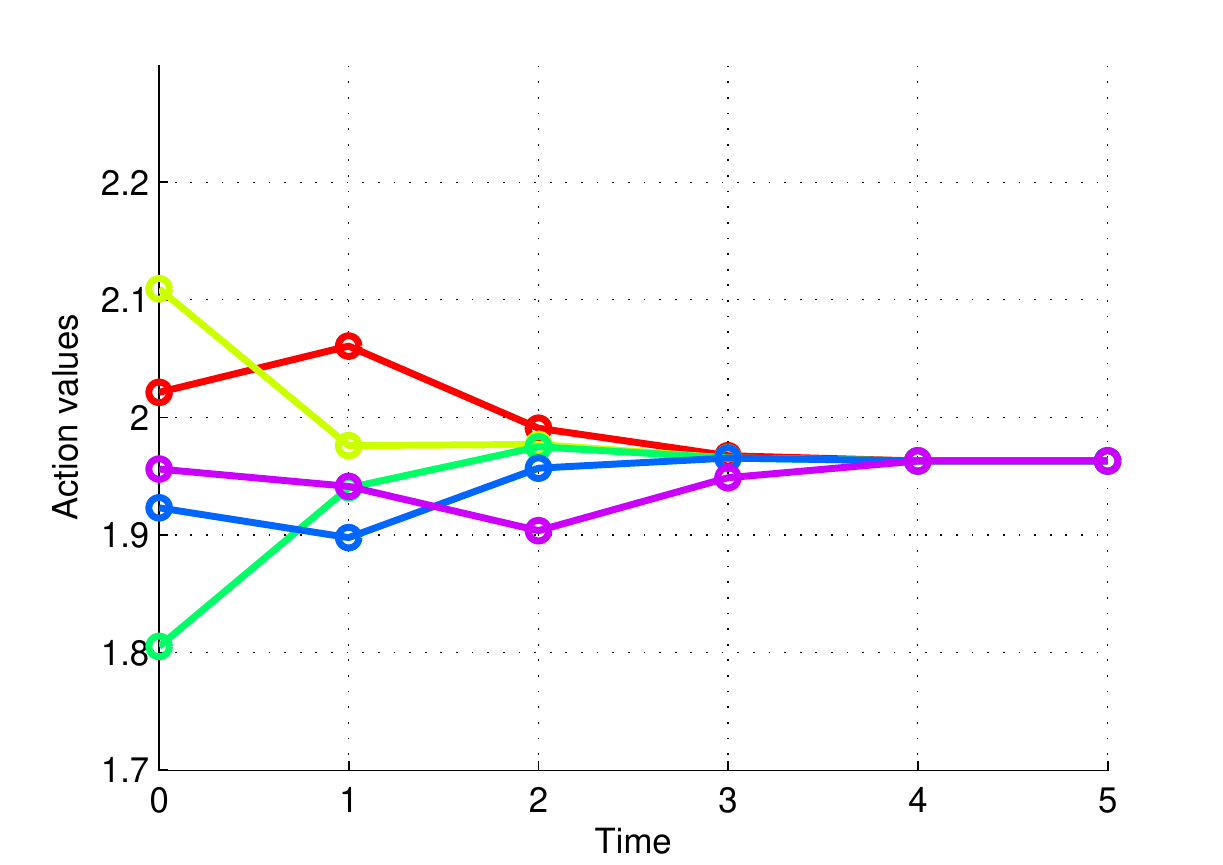} 
&\includegraphics[width=0.32\linewidth]{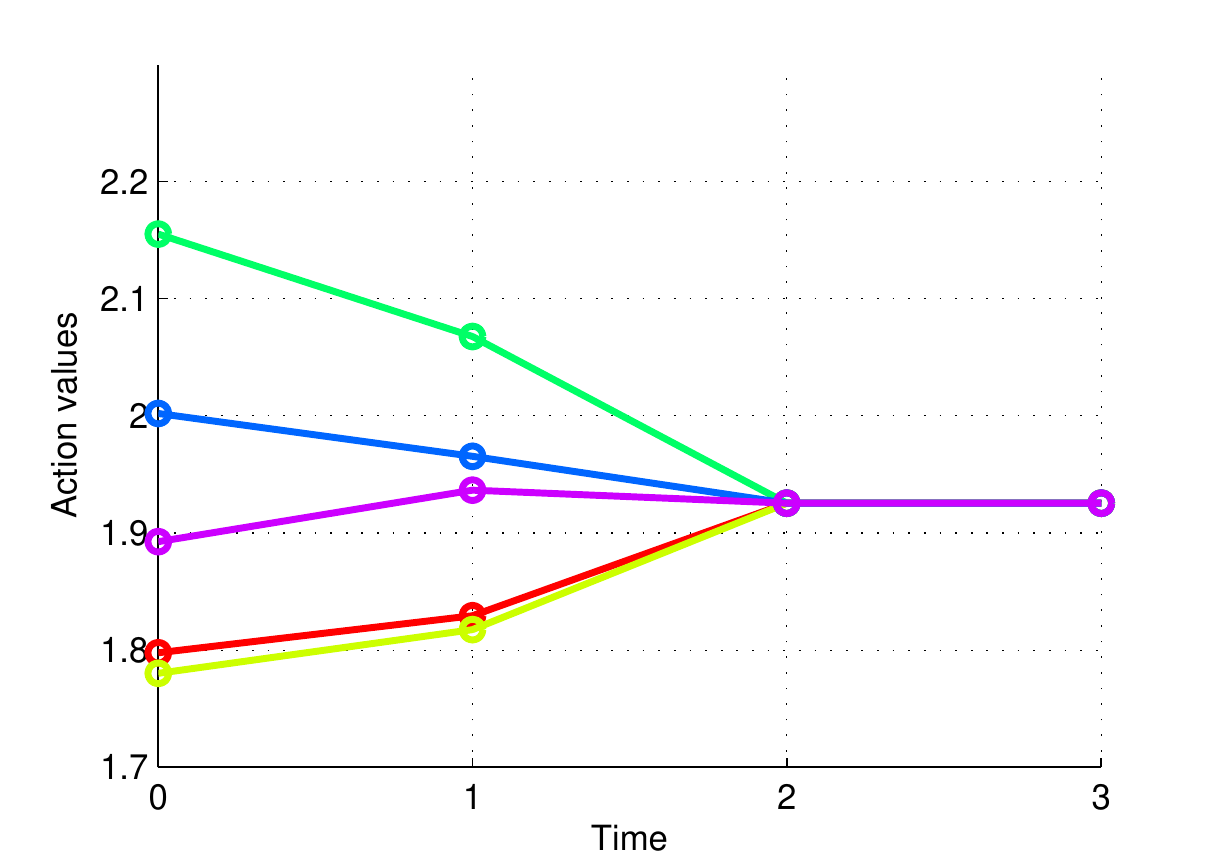}
&\includegraphics[width=0.32\linewidth]{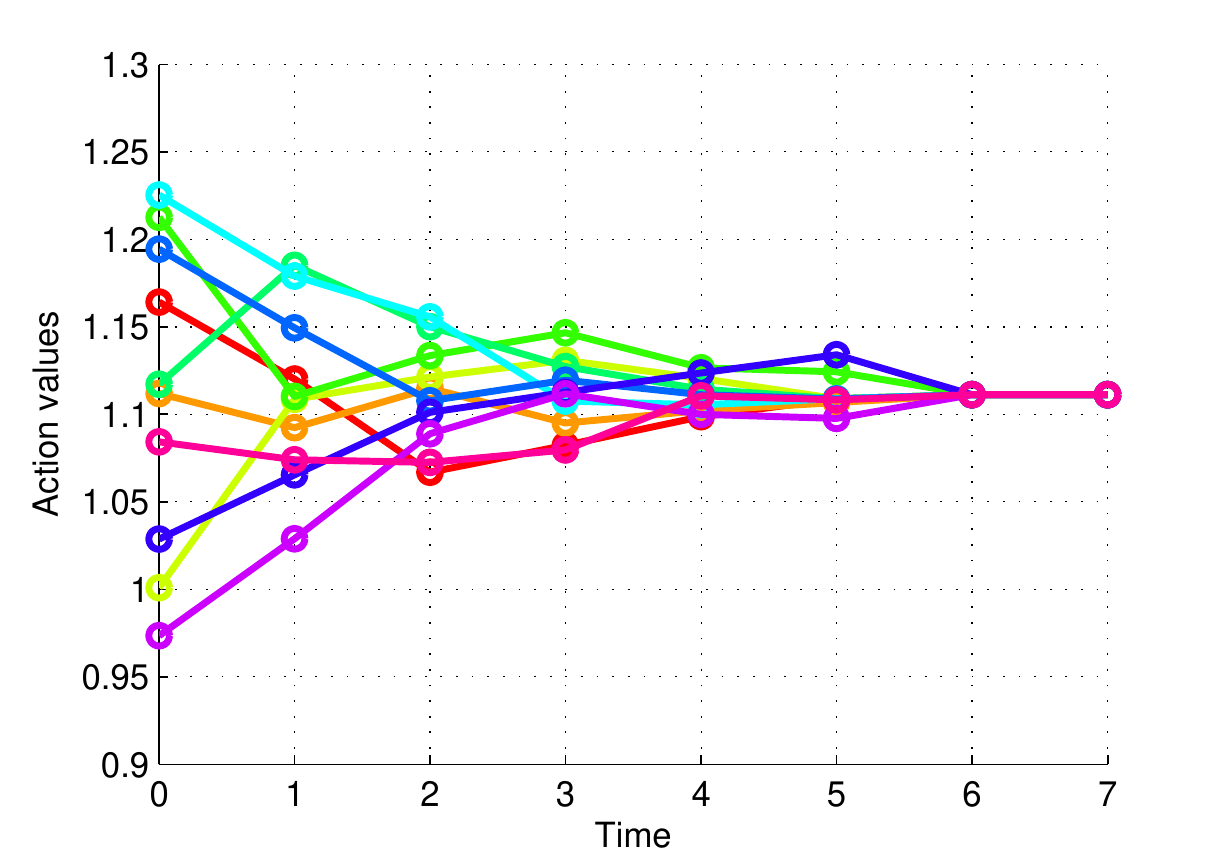} \\
\vspace{-2mm}
    \fontsize{7}{12}\selectfont (a)
      	       & \fontsize{7}{12}\selectfont (b)
 &     \fontsize{7}{12}\selectfont (c)
 \end{tabular}\vspace{-1mm}
\caption{Agents' actions over time for the Cournot competition game and networks shown in Fig. \ref{line_star_ring}. Each line indicates the quantity produced for an individual at each stage. Actions converge to the Nash equilibrium action of the complete information game in the number of steps equal to the diameter of the network.} 
\label{network_results}
\end{figure}
%%%

\begin{figure}[!t]
\centering
\begin{tabular}{ccc} \hspace{-5mm}
\includegraphics[width=0.32\linewidth]{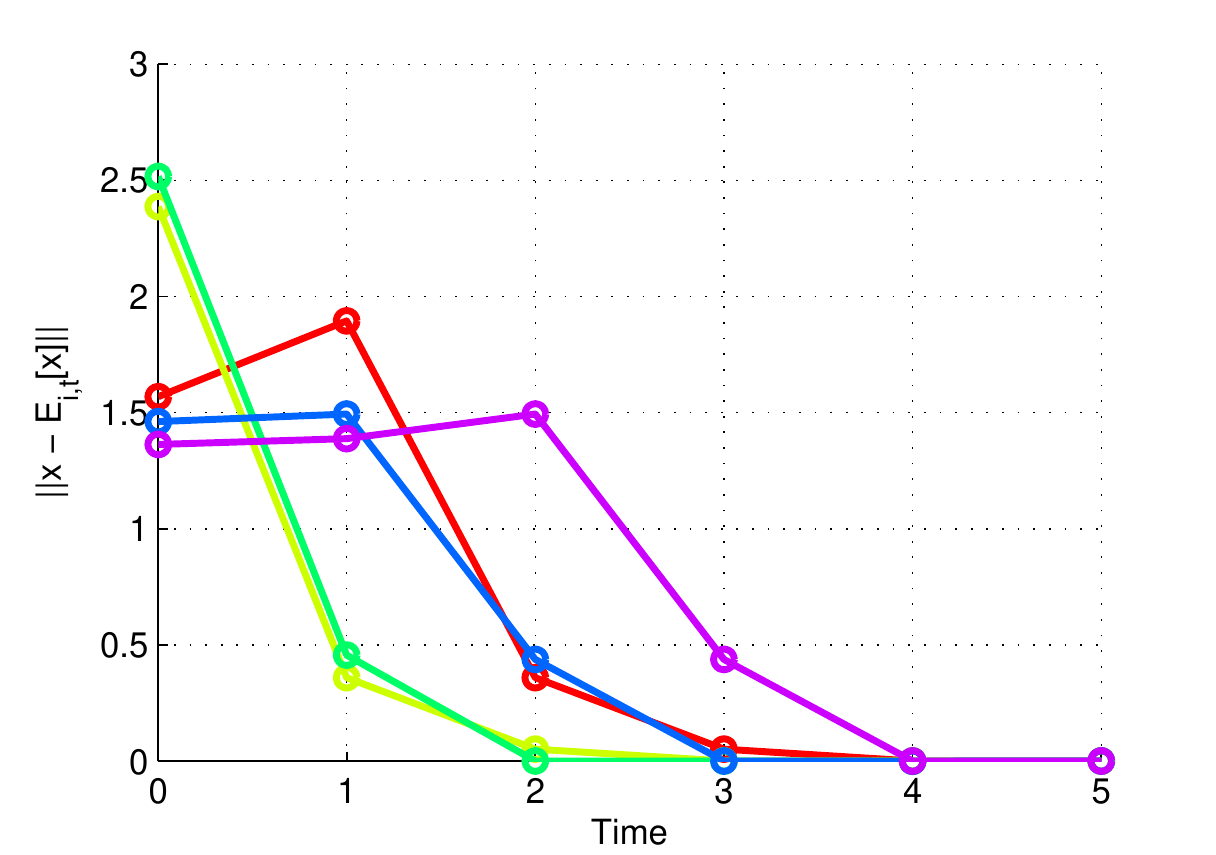}
&\includegraphics[width=0.32\linewidth]{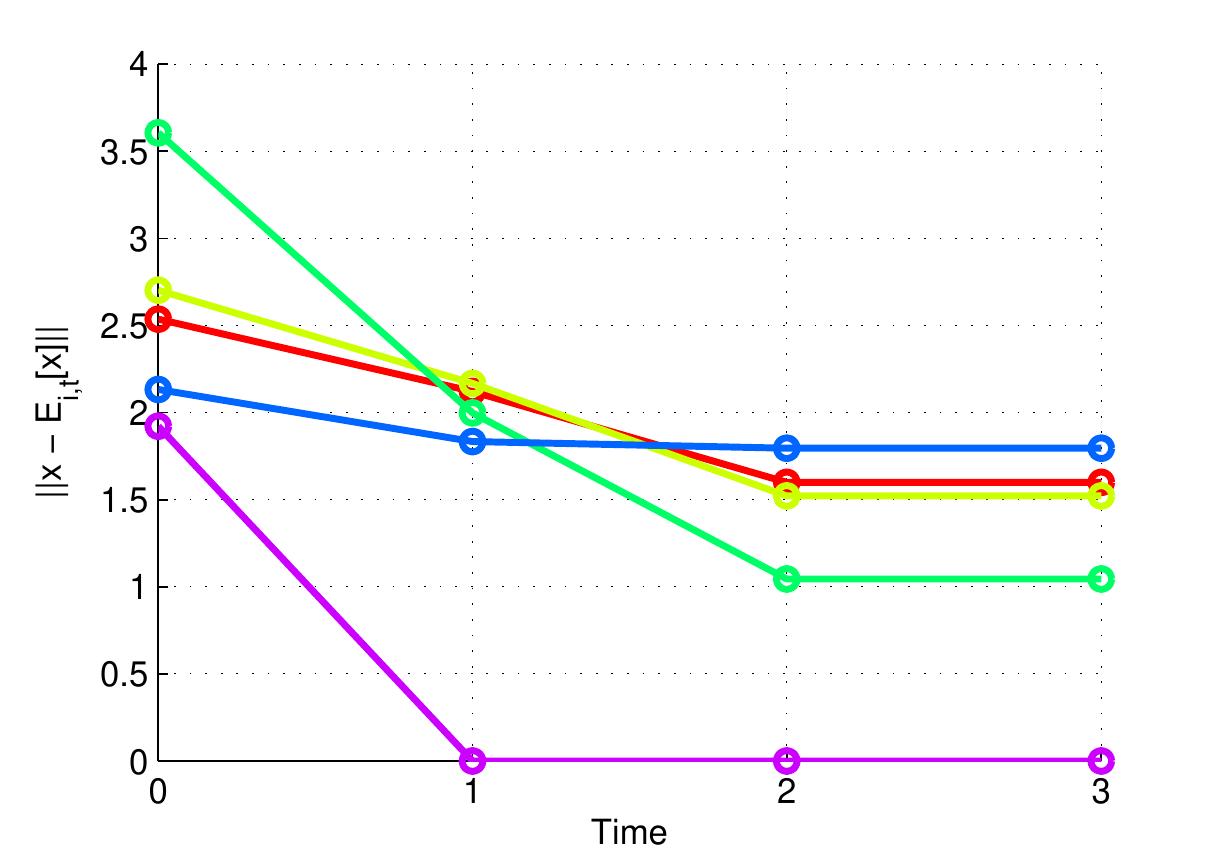}
&\includegraphics[width=0.32\linewidth]{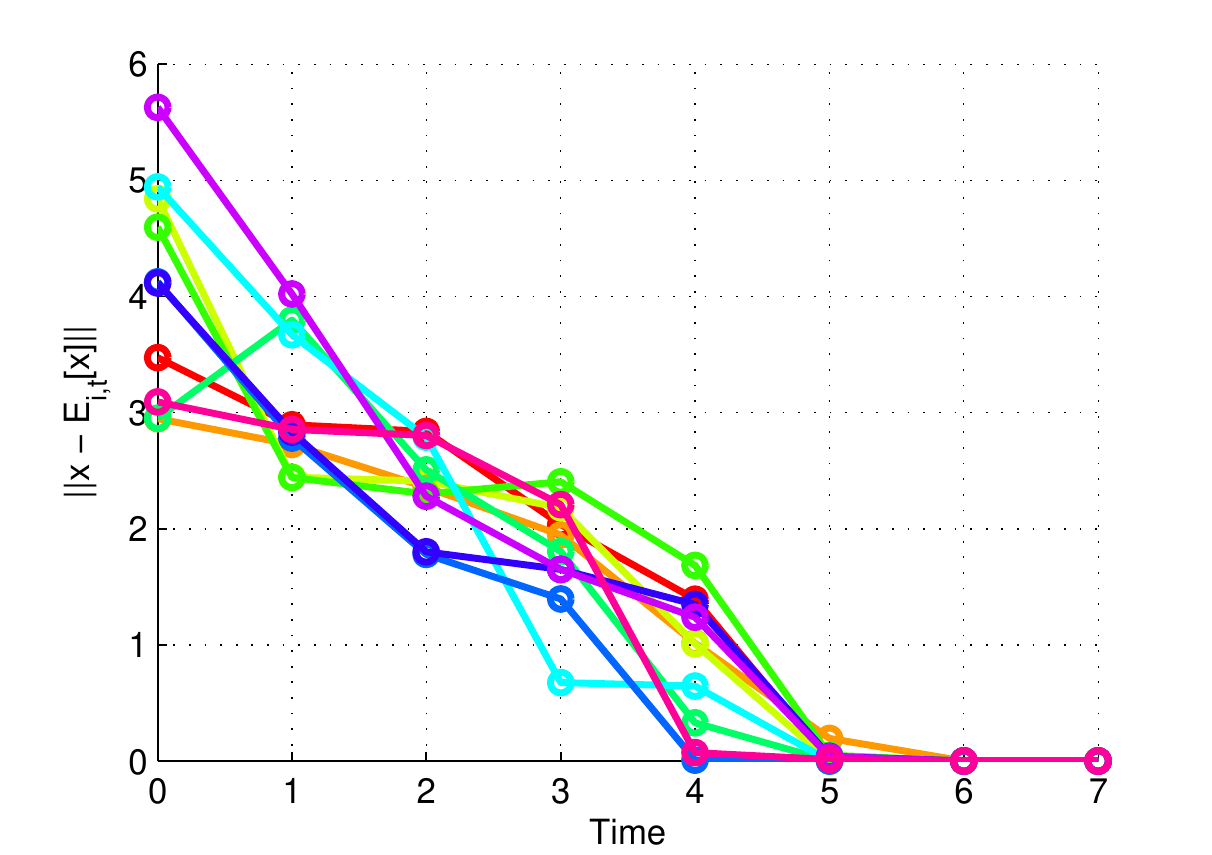} \\
\vspace{-2mm}
    \fontsize{7}{12}\selectfont (a)
      	       & \fontsize{7}{12}\selectfont (b)
 &     \fontsize{7}{12}\selectfont (c) \vspace{-1mm}
 \end{tabular}
\caption{Normed error in estimates of privates signals, $\|\bbx - \bbE_{i,t}[\bbx]\|_2^2$, for the Cournot competition game and networks shown in Fig. \ref{line_star_ring}. Each line corresponds to an agent's normed error in mean estimates of private signals over the time horizon. While all of the agents learn the true values of all the private signals in line and ring networks, in the star network only the central agent learns all of the private signals. }%However, as we observe in Fig. \ref{network_results}, all of the agents learn the sufficient statistic for computing the estimate of $\theta$ given $\bbx$.}
\vspace{-3mm}
\label{network_beliefs}
\end{figure}

\subsection{Learning in Cournot competition} \label{Cournot_num_example}

The underlying effective unit profit is chosen as $\theta = 12\$/$unit. Firms observe private signals with the additive noise term coming from standard normal distribution, i.e., $\epsilon_i \sim \ccalN(0,1)$. Given this setting, we consider three benchmark networks: a line network with $N=5$ firms,  a star network with $N=5$ firms, and a ring  network with $N = 10$ firms (see Fig. \ref{line_star_ring}). 

The quantities produced by firms over time are shown in Fig. \ref{network_results} for the line (a), star (b) and ring (c) networks. In all of the cases, we observe consensus in the units produced. Furthermore, the consensus production $a^*$ is optimal; that is, firms converge to the Bayes-Nash equilibrium under complete information \eqref{Cournot_Nash}. This implies that all of the firms learn the best estimate of $\theta$ by the convergence time $T$, that is, $ \bbE_{i,T}[\theta \given h_{i,T}] =  \bbE[\theta \given \bbx]$ for all $i \in V$. 

Figs. \ref{network_beliefs}(a)--(c) show the error in estimation of private signals $\|\bbx - \bbE_{i,t}[\bbx]\|^2_2$ for all $i \in V$ and $t \in \naturals$. In Figs. \ref{network_beliefs}(a) and \ref{network_beliefs}(c), corresponding to line and ring networks, the mean square error in private signal estimates goes to zero for all of the firms at the end of the convergence time $T$. On the other hand, in the star network in Fig. \ref{network_beliefs}(b), except for the center firm $5$, none of the other firms has zero mean square error in private signal estimates. This means that these firms do not learn at least one of the private signals. As we know from Fig. \ref{network_results} (b), all of the firms in the star network learn the best estimate of $\theta$ given all of the private signals. Hence, in the star network, firms only learn the sufficient statistic to estimate $\theta$ (which is the average of the private signals) rather than learning each of the private signals individually. 

Figs. \ref{network_results}(a)--(c) suggest that convergence is achieved in $O(\Delta)$ steps where $\Delta$ is the diameter of the graph. In \cite{mossel2010efficient}, it is argued that for the distributed estimation problems when the individual utility function is equal to $u_i(\bba_i, \theta) = -(\bba_i - \theta)^2$, convergence happens in $O(\Delta)$ steps for tree networks. Our results show that the convergence rate is $O(\Delta)$ not only for tree networks such as line and star networks but also for the ring network when the utility function is quadratic and includes actions of others. 
%!TEX root = bqng_arxiv.tex

\begin{figure}
\center
\tikzstyle{ellip} =[ellipse, draw=black, dashed, minimum width = 1.5cm, minimum height = 0.9cm,rotate=22]
\tikzstyle{agent} = [circle, draw=black, inner sep=0pt, minimum size=0.4cm, fill=blue!30]
\tikzstyle{arrow} = [stealth-stealth, thin]
{\tiny\begin{tikzpicture}[scale=1]
\path    ( 0,    0)                           node (1) [agent]  {1}
      ++ (-1.8,   -0.6)    + (0.1*rand, 0.1*rand) node (2) [agent]  {2}
      ++ ( 1,   -0.7)                           node (3) [agent]  {3}
      ++ ( 2.5,  1)  + (0.1*rand, 0.2*rand) node (4) [agent]  {4}
      ++ ( 0,   -0.8)    + (0.1*rand, 0.1*rand) node (5) [agent]  {5}
      ++( -1.7, 1.1)     node (6) [ellip]  {};

\draw [arrow, black!20] (1)  to (2);
\draw [arrow, black!20] (1)  to (4);
\draw [arrow, black!20] (2)  to (3);
\draw [arrow, black!20] (3)  to (5);
\draw [arrow, black!20] (4)  to (5);

\draw [-stealth, red, thick] (1) -- ++(0.7,+0.4);
\draw [-stealth, red, thick] (2) -- ++(1,+0.3*rand);
\draw [-stealth, thick] (2) -- ++(0.8,+0.4) ;
\draw [-stealth, red, thick] (3) -- ++(1,+0.3*rand);
\draw [-stealth, thick] (3) -- ++(0.5,+0.8) ;
\draw [-stealth, red, thick] (4) -- ++(1,+0.3*rand);
\draw [-stealth, thick] (4) -- ++(0.4,+0.8) ;
\draw [-stealth, red, thick] (5) -- ++(1,+0.3*rand);
\draw [-stealth, thick] (5) -- ++(0.6,+0.8) ;
\draw [-stealth, dashed] (1) -- ++(-0.6, -0.4);
\draw [-stealth, dashed] (1) -- ++(-0.4, 0.4);
\draw [-stealth, dashed, thin] (1) -- ++(0,+0.9);
\draw [-stealth, dashed, thin] (1) -- ++(0,-0.9);
\draw [-stealth, thick] (1) -- ++(0.6,+1) ;

\node [left] at (0,0.9) {$z$};
\node [left] at (-0.4,0.5) {$x$};
\node [left] at (-0.6,-0.4) {$y$};

\end{tikzpicture}}\vspace{-1mm}
\caption{Mobile agents in a 3-dimensional coordination game.  Agents observe initial noisy private signals on heading and take-off angles. Agents revise their estimates on true heading and take-off angles and coordinate their movement angles with each other through local observations.}\vspace{-5mm}
 \label{coordination_picture}
\end{figure}
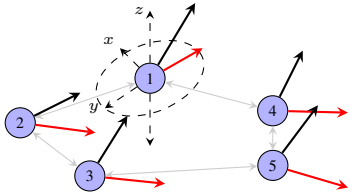

%%%%%%%%%%%%%%%%%%%%%%%%%%%%%%%%%
%%%%%%%%%%%%%%%%%%%%%%%%%%%%%%%%%
%%%%%%%%  E X A M P L E  %%%%%%%%%%%%%%%%%%
%%%%%%%%%%%%%%%%%%%%%%%%%%%%%%%%%
%%%%%%%%%%%%%%%%%%%%%%%%%%%%%%%%%
\section{Coordination Game} \label{coordination_game}
%
%\blue{This section is not completely compatible with changes we made in other parts of the paper. It also needs a writing polishing pass.}

A network of autonomous agents want to align themselves so that they move toward a goal $(x^*,y^*,z^*)$ on 3-dimensional space following a straight path, and at the same time maintain their initial starting formation. When the goal $(x^*,y^*,z^*)$ is far away, then there exists a common correct direction of movement  toward the goal characterized by the heading angle on the $x-y$ plane $\phi \in [0^\circ, 180^\circ]$ and the take-off angle on the $x-z$ plane $\psi \in [0^\circ, 180^\circ]$. Hence, the target movement direction is given by $\bbtheta = [\phi, \psi]^T$. Fig. \ref{coordination_picture} illustrates a set of autonomous agents on a 3-dimensional plane and their heading and take-off angles where the $x$, $y$, $z$ axes are depicted for agent $1$. 

Mobile agents have the goal of maintaining the starting formation while moving at equal speed by coordinating their movement direction with other agents. Agents need to coordinate with the entire population while communication is restricted to neighboring agents whose direction of movement they can observe.
In this context, agent $i$'s decision $\bba_i \in [0, 180^\circ] \times [0, 180^\circ]$ represents the heading and take-off angles in the direction of movement. The estimation and coordination goals of agent $i$ can be represented with the following payoff
\begin{align}\label{Beauty_Contest_cost}
u_i(\bba_i, \{\bba_{j}\}_{j \in V\setminus i}, \bbtheta) &  =  - \frac{1-\lambda}{2}(\bba_i - \bbtheta)^T (\bba_i - \bbtheta)
-\:\frac{\lambda}{2(N-1)}\hspace{-0.2cm} \sum_{ j \in V \setminus \{i\}}\hspace{-0.3cm}(\bba_i - \bba_j)^T(\bba_i - \bba_j). 
\end{align}
The first term is the estimation error in the true heading and take-off angles. The second term is the coordination component that measures the discrepancy between the direction of movement and those of other agents. %The utility function is reminiscent of the one in a potential game \cite{marden2009cooperative}. 
$\lambda$ is a constant in $(0,1)$ gauging the importance of estimation term with respect to the coordination term. 

The same payoff formulation can be motivated by looking at learning in organizations \cite{ArmengolMarti}. In an organization, individuals share a set of common tasks and have the incentive to coordinate with other units. Each individual receives a private piece of information about the task that needs to be performed while only being able to share his information with whom he has a direct contact in the organization.

Note that the utility function is of the quadratic form given in \eqref{quadratic_utility_form_vector} with vector states and vector actions. Hence, we can use the QNG filter in Section \ref{vector_state} as summarized in Algorithm \ref{alg2}. %Each agent's best response to the strategies of others $\{\sigma_{j,t}\}_{j \in V \setminus i}$ given the information available to his at time $t$ is obtained by solving  $\partial \bbE_{i,t}[u_i(\bba_i, \{\sigma_{j,t}(h_{j,t})\}_{j \in V \setminus i}, \bbtheta)]/ \partial \bba_i = 0$.
As postulated in \eqref{linear_reply}, the explicit equilibrium equation for all $i\in V$ is
\begin{align}\label{BR_beauty}
\sigma^*_{i,t}(h_{i,t})  = (1-\lambda)\bbE_{i,t}[ \bbtheta] +\frac{\lambda}{N-1} \hspace{-2mm}\sum_{j \in V \setminus \{i\}} \hspace{-2mm} \bbE_{i,t}[ \sigma^*_{j,t}(h_{j,t}))].
\end{align}

In a game with complete information, the Bayes-Nash equilibrium actions of all agents coincide and are given by
\begin{equation} \label{Coordination_Nash}
   a^*_i =\bbE[\bbtheta \given \bbx].
\end{equation}
%

%\section{Numerical examples} \label{simulation}
%
%
%Here, we consider the two models corresponding to examples in Sections \ref{cournot} and \ref{coordination_game}. Private signals are Gaussian with mean equal to the underlying state of the world $\theta \in \reals^d$ and variance equal to $\bbI_{d \times d}$. When $d=1$, agents follow the equilibrium behavior prescribed by equations in figs. \ref{feedback_diagram} and \ref{gains_diagram}. In the case where $d>1$, agents follow Algorithm \ref{alg2}. 

In the next section, we show that the equilibrium actions in \eqref{BR_beauty} converge to the Bayes-Nash equilibrium with complete information as given by \eqref{Coordination_Nash} in finite number of steps. 

\subsection{Learning in coordination games}

The correct direction vector is chosen to be $\theta = [10^\circ, 20^\circ]^T$. We let $\lambda = 0.5$. The noise terms, $\bbepsilon_i$ are jointly Gaussian with mean zero and covariance matrix equal to the identity matrix. Having an identity covariance matrix implies that $\bbE[\bbx_i[1]\bbx_i[2]] = 0$. 

We evaluate equilibrium behavior in geometric and random networks with $N= 50$ agents, Figs. \ref{large_networks} (a) and (b), respectively. Geometric random network is created by placing the agents randomly on a $4$ meter  $\times$ 4 meter square and connecting pairs with distance less than $1$ meter between them. In the random network, any pair of agents are neighbors with probability 0.1. The geometric network in Fig. \ref{large_networks} (a) has a diameter of $\Delta_g = 5$ where the random network in Fig. \ref{large_networks} (b) has a diameter of $\Delta_r = 4$. 

The direction of movement of each agent over time is depicted in Figs. \ref{large_network_results}(a)--(d). Figs. \ref{large_network_results}(a) and \ref{large_network_results}(b) show the heading angle $\phi_i$ of agents in geometric and random networks, respectively. Figs. \ref{large_network_results}(c) and \ref{large_network_results}(d) show the take-off angle $\psi_i$ of agents in geometric and random networks, respectively. Fig. \ref{large_network_results} illustrates that agents' movement directions converge to the best estimates in heading and take-off angles in a finite number of steps.  As a result, at the end of the convergence time $T$, we have $\bbE_{i,t} [\phi \given h_{i,T}] = \bbE[\phi \given \bbx[1]]$ and $\bbE_{i,t} [\psi \given h_{i,T}] = \bbE[\psi \given \bbx[2]]$ for all $i \in V$. Further, convergence time is in the order of the diameter for both of the networks. This means that agents learn the sufficient statistic to calculate best estimates in the amount of time it takes for information to propagate through the network.

%Action consensus is reasonable since agents have the incentive to agree with others in the movement direction in order to maintain the initial formation---see \eqref{Beauty_Contest_cost}. Further, convergence takes place in a finite number of steps in the order of the diameters for both networks. %These examples support the results for the scalar case shown in Section \ref{Cournot_num_example} for generic networks. 

\begin{figure}[!t]
\centering
\begin{tabular}{cc} \hspace{-4mm}
\includegraphics[width=0.42\linewidth, height=0.3\linewidth]{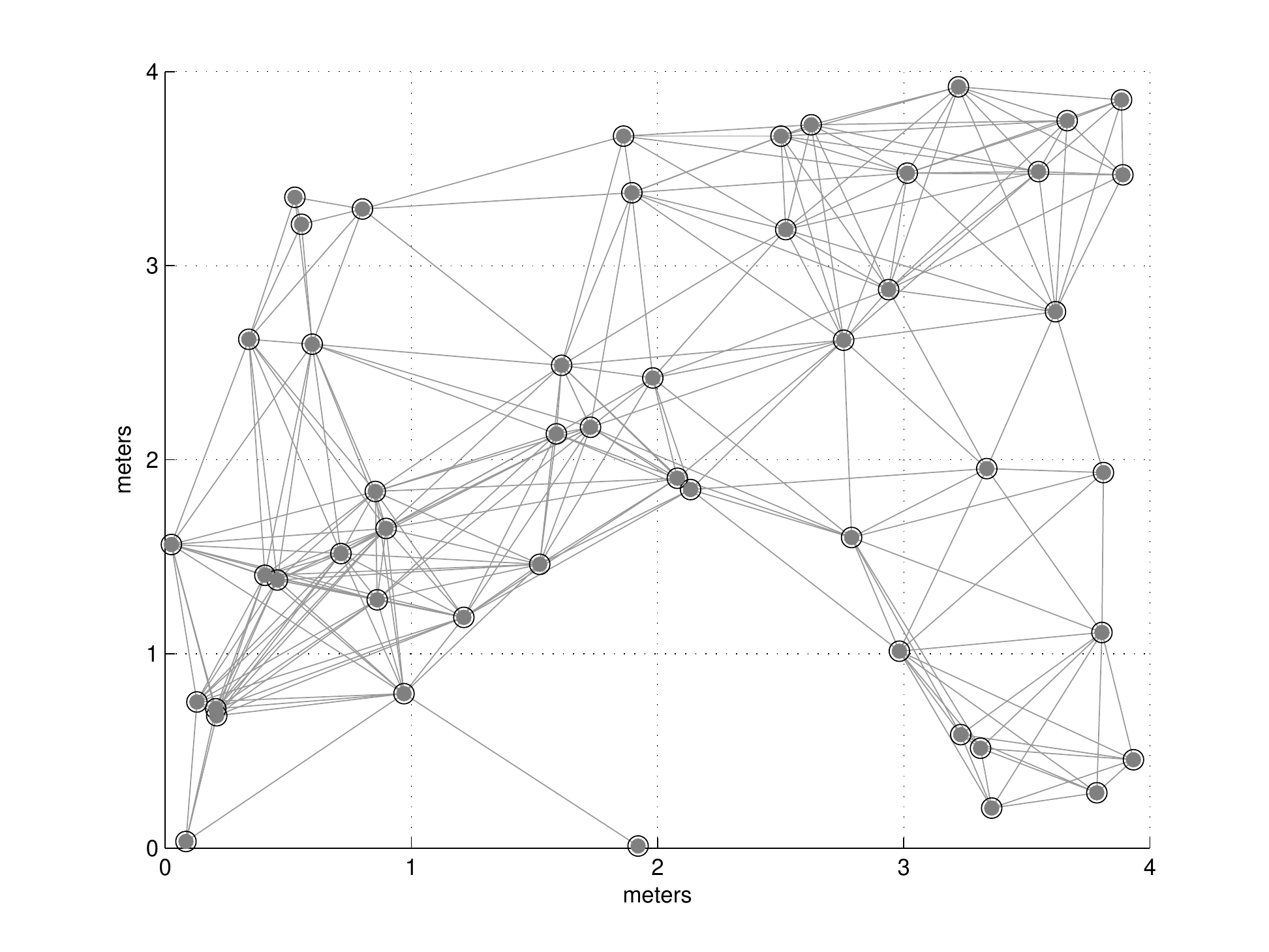}
&\includegraphics[width=0.42\linewidth, height=0.3\linewidth ]{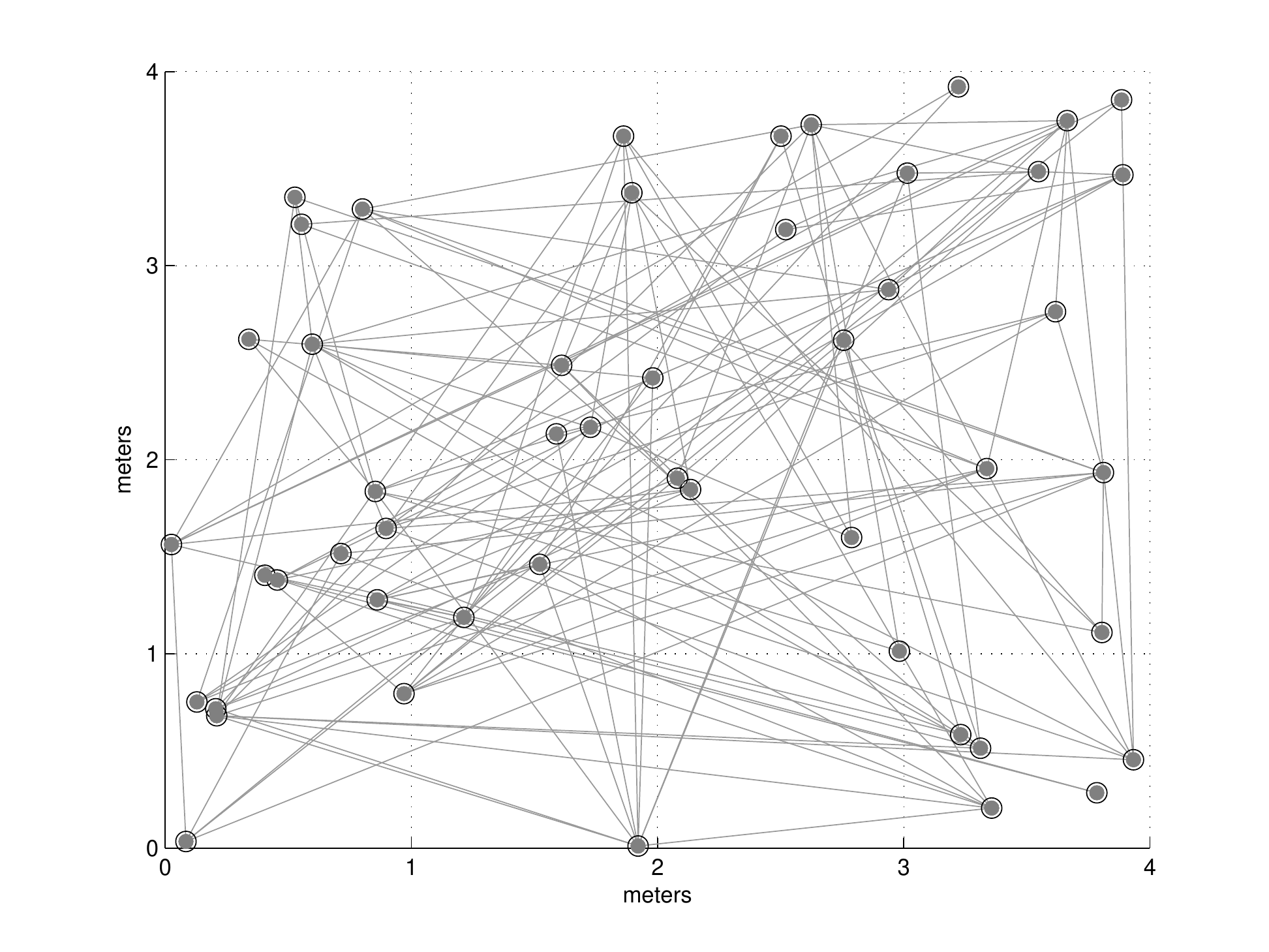} \vspace{-2mm}
 \\
    \fontsize{7}{12}\selectfont (a)
      	       & \fontsize{7}{12}\selectfont (b)
 \end{tabular}\vspace{-1mm}
\caption{Geometric (a) and random (b) networks with $N = 50$ agents. Agents are randomly placed on a $4$ meter $\times$ $4$ meter square. There exists an edge between any pair of agents with distance less than $1$ meter apart in the geometric network. In the random network, the connection probability between any pair of agents is independent and equal to $0.1$. }
\label{large_networks}\vspace{-3mm}
\end{figure}

\begin{figure}[!t]
\centering
\begin{tabular}{cc}
\includegraphics[width=0.42\linewidth, height=0.3\linewidth ]{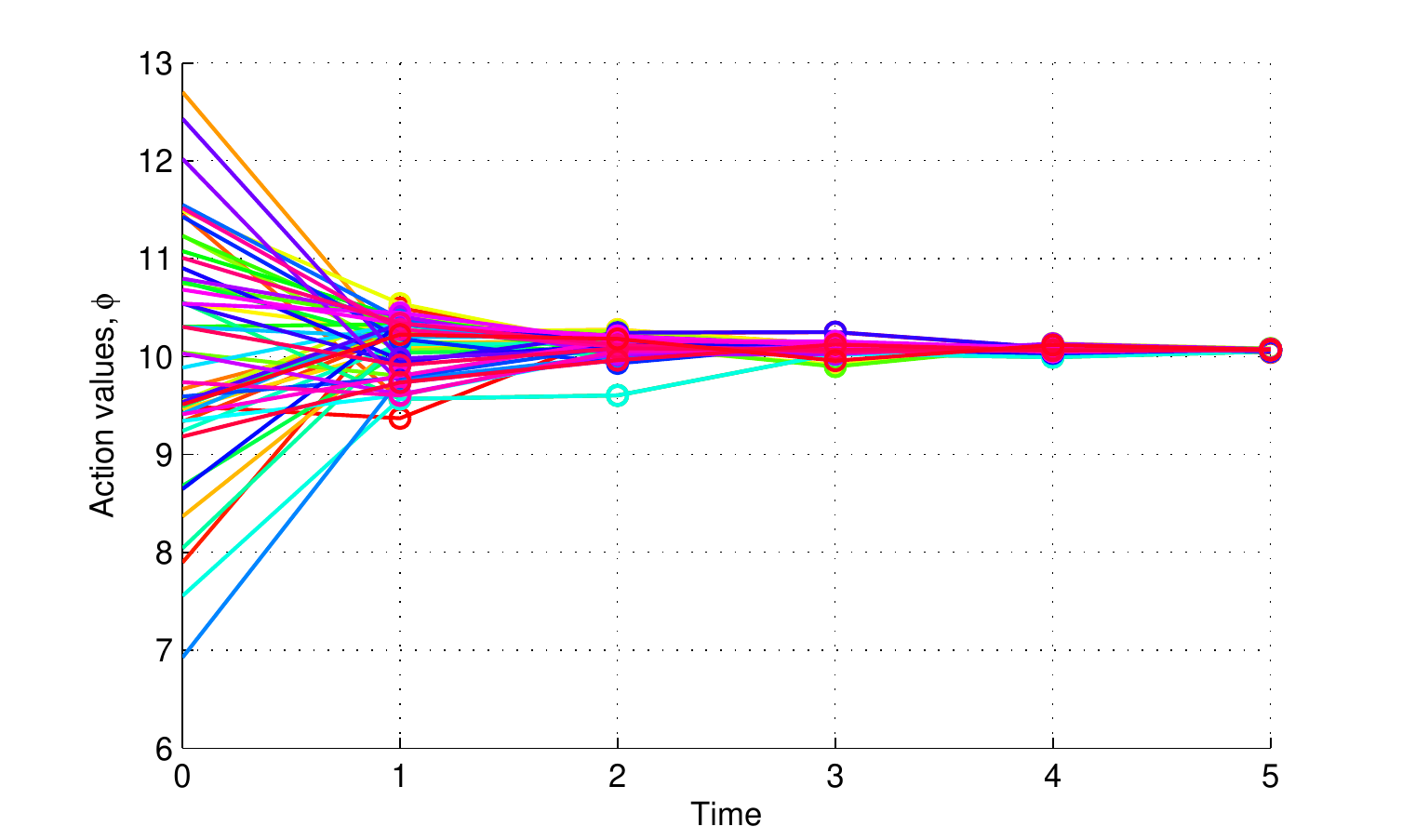}
&\includegraphics[width=0.42\linewidth, height=0.3\linewidth ]{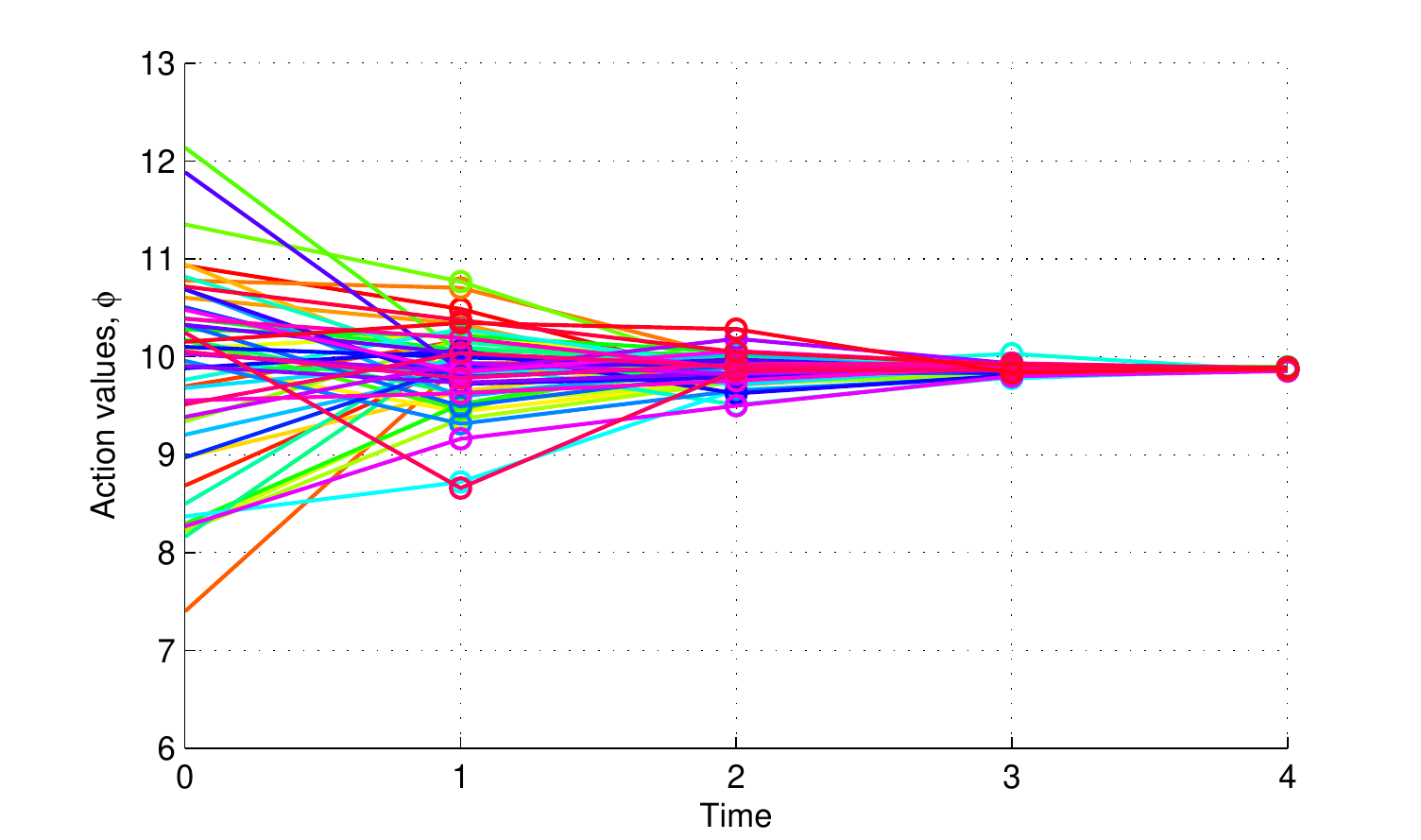} \vspace{-2mm}
 \\
    \fontsize{7}{12}\selectfont (a)
      	       & \fontsize{7}{12}\selectfont (b) \\
\includegraphics[width=0.42\linewidth, height=0.3\linewidth ]{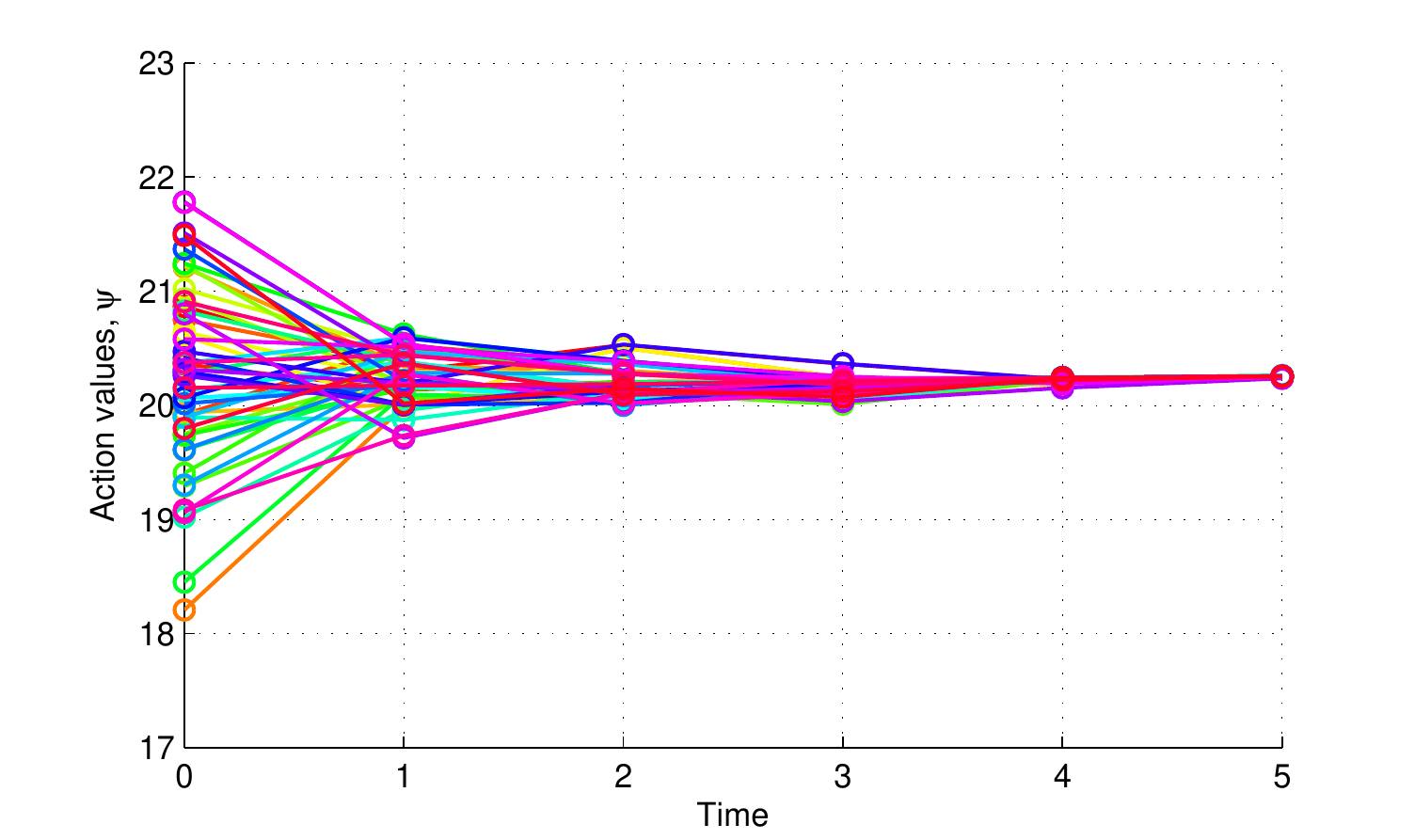}
&\includegraphics[width=0.42\linewidth, height=0.3\linewidth ]{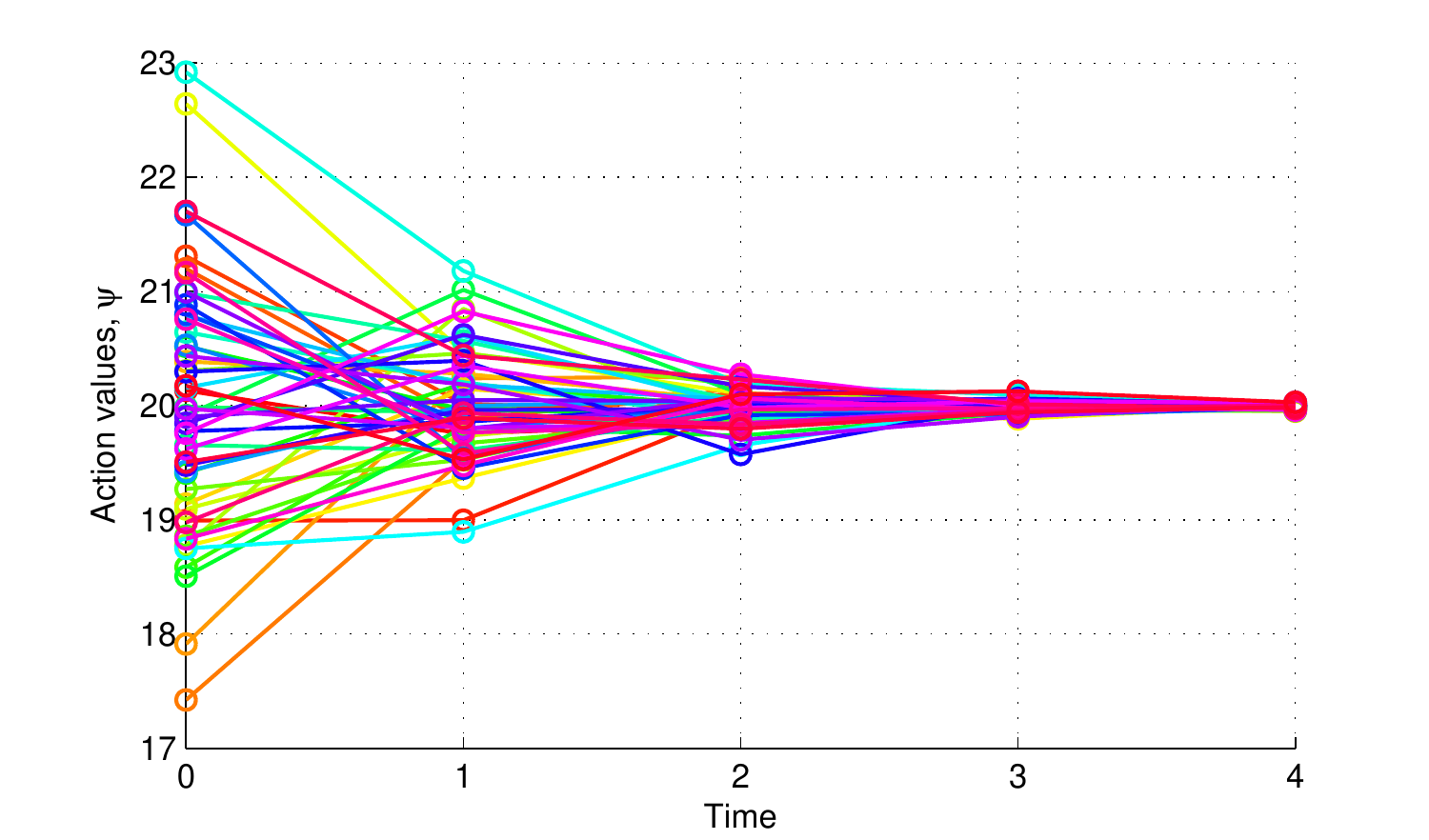}  \vspace{-2mm} \\
    \fontsize{7}{12}\selectfont (c)
      	       & \fontsize{7}{12}\selectfont (d)
 \end{tabular}\vspace{-1mm}
\caption{
Agents' actions over time for the coordination game and networks shown in Fig. \ref{large_networks}. Values of agents' actions over time for heading angle $\phi_i$ (top) and take-off angle $\psi_i$ in geometric (left) and random (right) networks respectively. Action consensus happens in the order of the diameter of the corresponding networks.} \vspace{-3mm}
\label{large_network_results}
\end{figure}

\section{Conclusion} \label{conclusion}

In this paper we introduced the QNG filter that agents can run locally to update their beliefs and select equilibrium actions actions in repeated quadratic games with both information and payoff externalities. The QNG filter provides a mechanism to update beliefs in a Bayes' way when agents' initial prior over the state of the world is Gaussian. We began by showing that when the prior estimates of private signals are Gaussian with means equal to a linear combination of private signals, and the equilibrium strategies of agents are linear combination of mean estimates of private signals, Bayesian updates of estimates of private signals and the underlying state follow a sequential LMMSE estimator. This meant that the estimates remain linear combinations of private signals, and hence, Gaussian. By induction, estimates remain Gaussian for all times if equilibrium actions that are linear in mean of the estimates exist at all the stages. Further, we derived an explicit recursion for tracking of estimates of private signals and calculating equilibrium actions which we leverage to develop the QNG filter. We then extended the QNG filter to the case when the state of the world is a vector. We exemplified the QNG filter in Cournot competition game and coordination of mobile agents on 3-dimensional space. In the former the state of the world, effective profit, was a scalar, whereas in the latter the state of the world was a vector including heading and take-off angles. In both examples, the QNG filter converged to the BNE of the game in number of steps that is equal to the order of the diameter of the network. This meant that agents learnt the sufficient statistic of the state while not necessarily learning all the individual private signals. 

\begin{appendices}
\section{Proof of Theorem \ref{rational_updates_theorem}:} \label{Theorem_1_proof}
At time $t=0$ beliefs are normal and have the form in \eqref{linear_estimates_first_one}. Indeed, since the only information available to agent $i$ at time $t=0$ is the private signal $x_i$ it follows from the linear observation model in \eqref{private_signal_model} that this is the value assigned to the estimate of all private signals as well as to the estimate of the state $\theta$,
\begin{align} \label{initial_mean_estimates}
   \bbE_{i,0} \left[x_j\right]   = x_{i} \text{\ for all\ } j,\quad
   \bbE_{i,0} \left[\theta\right] = x_{i}. 
\end{align}
The elements of the matrix $L_{i,0} = \bbone\bbe_i^T$ are $1$ in the $i$th column and $0$ otherwise. Therefore, the first expression in \eqref{initial_mean_estimates} is equivalent to the first expression in \eqref{initial_weights}. Likewise, since the $i$th element of $\bbe_i$ is one with remaining elements zero, the second expression in \eqref{initial_mean_estimates} is equivalent to the second expression in \eqref{initial_weights}. As for the variances in \eqref{initial_variances}, note that the initial estimate of $\bbx$ has error covariance matrix defined as in \eqref{error_covariance_xx} for $t=0$. By substituting initial mean estimates inside \eqref{error_covariance_xx} and then using the fact that $\bbe_i^T \bbx = x_i$, the error covariance matrix can be rewritten as
\begin{align}
     M^i_{\bbx \bbx}(0) 
%          =& \bbE_{i,0}\Big[\big(\bbx - \bbE_{i,0} 
%               \left[\bbx\right]\big)\big(\bbx - \bbE_{i,0} \left[\bbx\right]\big)^T\Big]  \label{initial_covariance_estimates}\\
%           =& \bbE_{i,0}\Big[\big(\bbx - \bbone\bbe_i^T\bbx\big)
%                            \big(\bbx - \bbone\bbe_i^T\bbx \big)^T\Big] \label{initial_covariance_estimates_1} \\
            =& \bbE_{i,0}\Big[\big(\bbx - \bbone x_i\big)
                            \big(\bbx - \bbone x_i\big)^T\Big] \label{initial_covariance_estimates_2}
%           =& \bbE_{i,0}\Big[\big(\bbepsilon - \bbone \eps_i\big)
%                            \big(\bbepsilon - \bbone \eps_i\big)^T\Big]. \label{initial_covariance_estimates_3}
\end{align}
%
%We get the equality \eqref{initial_covariance_estimates_2} by substituting initial mean estimates inside \eqref{initial_covariance_estimates} and then using the the fact that $\bbe_i^T \bbx = x_i$.
 From \eqref{initial_covariance_estimates_2}, we get the following by using the fact that $x_j - x_i = \epsilon_j -\epsilon_i$ by \eqref{private_signal_model},
\begin{align}
     M^i_{\bbx \bbx}(0) =& \bbE_{i,0}\Big[\big(\bbepsilon - \bbone \eps_i\big)
                            \big(\bbepsilon - \bbone \eps_i\big)^T\Big]. \label{initial_covariance_estimates_3}
\end{align}
  When we expand the terms in \eqref{initial_covariance_estimates_3}, we obtain the following           
\begin{align}
M^i_{\bbx \bbx}(0) 
           =& \bbE_{i,0}\left[\bbepsilon \bbepsilon^T\right] - \bbE_{i,0}\left[\bbepsilon \bbone^T \epsilon_i\right] 
          -\bbE_{i,0}\left[\bbone \epsilon_i \bbepsilon^T \right]  + \bbone \bbone^T \bbE_{i,0}\left[ \epsilon_i^2 \right] \label{initial_covariance_estimates_4} \\
            =& \text{diag}(\bbc) - \bbe_i \bbone^T c_i - \bbone \bbe_i^T c_i + \bbone \bbone^T c_i \label{initial_covariance_estimates_5}\\
            =& \text{diag}(\bbc) +  \bar\bbe_i\bar\bbe_i^T c_i - \bbe_i \bbe_i^T c_i  \label{initial_covariance_estimates_6}
\end{align}
Since private signals are independent among agents, that is $\bbE_{i,0}[\epsilon_k \epsilon_j]  = 0$ for all $j \in V \setminus k$ and $k \in V$, we have $\bbE_{i,0}[\bbepsilon \bbepsilon^T] = \text{diag}(\bbc)$, $\bbE_{i,0}[\bbepsilon \epsilon_i] = \bbe_i c_i$. Using these relations  and the definition of noise variance $c_i = \bbE[ \epsilon_i^2]$, \eqref{initial_covariance_estimates_5} follows from \eqref{initial_covariance_estimates_4}.
When second and third terms are subtracted from the fourth term in \eqref{initial_covariance_estimates_5}, we obtain the last two terms in \eqref{initial_covariance_estimates_6}. Now, observe that $\text{diag}(\bbc) - \bbe_i \bbe_i^T c_i = \text{diag}(\bar\bbe_i) \diag(\bbc)$, hence \eqref{initial_covariance_estimates_6} can be rewritten as in \eqref{initial_variances}. 

Consider the variance of $\theta$ defined in \eqref{error_covariance_state}  at time $t =0$. Substituting $\bbE_{i,0}[\theta] = x_i$ inside \eqref{error_covariance_state}, we have 
\begin{align}
M^i_{\theta \theta}(0) = \bbE_{i,0}\left[ (\theta - x_i)^2 \right]
\end{align}
By the signal structure \eqref{private_signal_model} with additive zero mean Gaussian term $\epsilon_i$, we have $\theta- x_i = -\epsilon_i$. As a result, $M^i_{\theta \theta}(0) = \bbE_{i,0}[\epsilon_i^2]$ which is in return equal to $c_i$. Next consider the cross-covariance between $\theta$ and $\bbx$ defined in \eqref{error_covariance_thetax} at time $t = 0$,
\begin{align}
M^i_{\theta \bbx}(0) 
	=& \bbE_{i,0}\Big[\big(\theta - \bbE_{i,0}  \left[\theta\right]\big)\big(\bbx - \bbE_{i,0} \left[\bbx\right]\big)^T \Big] \label{initial_cross_covariance_estimates}\\
%	=& \bbE_{i,0}\Big[ (\theta - x_i) (x - \bbone x_i)^T \Big] \label{initial_cross_covariance_estimates_1} \\
	=& \bbE_{i,0}\Big[ (-\epsilon_i) (\bbepsilon - \bbone \epsilon_i)^T \Big]  \label{initial_cross_covariance_estimates_2}
%	=&(\bbone - \bbe_i)^T c_i.\label{initial_cross_covariance_estimates_2}
\end{align}
The second equality follows by substitution of initial mean estimates and then using the definition of private signals \eqref{private_signal_model}. Next, we multiply out the terms in \eqref{initial_cross_covariance_estimates_2}, use independence of private signals between agents to get \eqref{initial_variances}.

The inductive hypotheses is then true at time $t=0$ with the explicit initializations in \eqref{initial_weights} and \eqref{initial_variances}. Lemma \ref{rational_belief_updates_theorem_L} has already shown that if the inductive hypothesis is true at time $t$, it is also true at time $t+1$. It also provided the explicit recursions in \eqref{eqn_lmmse_gain_x}-\eqref{eqn_lmmse_gain_theta} and \eqref{weights_recursion_x}-\eqref{theta_x_covariance_estimate}. Lemma \ref{linear_equilibrium_at_time_t} further shows that the action coefficients $\bbv_{i,t}$ can be computed by solving the system of linear equations in \eqref{BNE_matrix_form}. %with $\bbv_t:= [\bbv_{1,t}^T, \ldots, \bbv_{N,t}^T]^T$, $\bbk_t :=  [\bbk_{1,t}^T, \ldots, \bbk_{N,t}^T]^T$ and $L_t$ as in \eqref{multipliers_matrix}.

\section{Proof of Theorem \ref{rational_updates_theorem_vector}:} \label{Theorem_2_proof}
At time $t =0$, agents beliefs are normal and have the form in \eqref{linear_estimates_first_one_vector}. Since the only information available to agent $i$ at time $t=0$ is the private signal $\bbx_i$, it follows from the observation model in \eqref{private_signal_vector} that agent $i$ assigns $\bbx_i$ as his mean estimates of the underlying parameter vector and the private signals as in \eqref{state_initial_estimate}-\eqref{signal_initial_estimate}. Next, consider the initial error covariance matrix $M^i_{\bbx \bbx}(0)$,
\begin{align}
M_{\bbx \bbx}^i(0) &= \bbE_{i,0}\left[ \left(\bbx - \bbE_{i,0}[\bbx]\right)\left( \bbx - \bbE_{i,0}[\bbx]\right)^T\right] \label{state_variance_initial_definition_vector}\\
%& = \bbE_{i,0}\left[ \left(\bbx - L_{i,0}\bbx\right)\left( \bbx - L_{i,0}\bbx\right)^T\right]\label{state_variance_initial_definition_vector_1} \\
&\hspace{-10mm} = \bbE_{i,0}\left[
 \left(\begin{array}{c}
                       \bbx[1] - \bbone \bbx_i[1] \\
                       \vdots\\
                        \bbx[N] - \bbone \bbx_i[N]
                     \end{array}\right) 
                      \left(\begin{array}{c}
                       \bbx[1] - \bbone \bbx_i[1] \\
                       \vdots\\
                        \bbx[N] - \bbone \bbx_i[N]
                     \end{array}\right)^T 
                      \right] \label{state_variance_initial_definition_vector_2}
%& = \bbE_{i,0}\left[
% \left( \begin{psmallmatrix} 
%                       (\bbx[1] - \bbone \bbx_i[1]) (\bbx[1] - \bbone \bbx_i[1])^T & (\bbx[1] - \bbone \bbx_i[1]) (\bbx[1] - \bbone \bbx_i[1])^T &\ldots &   (\bbx[1] - \bbone \bbx_i[1]) (\bbx[N] - \bbone \bbx_i[N])^T  \\
%                    \vdots & \cdots & \ddots &\vdots \\
%                        (\bbx[N] - \bbone \bbx_i[N])(\bbx[1] - \bbone \bbx_i[1] )^T & \ldots &
%                        \bbx[N] - \bbone \bbx_i[N]
%                     \end{psmallmatrix} \right)^T 
%                      \right]                    
\end{align}
Substituting initial mean estimates \eqref{signal_initial_estimate} in \eqref{state_variance_initial_definition_vector} and using the fact that  $\bbone \bbe_i^T \bbx[n] = \bbone \bbx_i[n]$, we get \eqref{state_variance_initial_definition_vector_2}. 
%The second equality follows by \eqref{signal_initial_estimate}. To get the last equality, we use the fact that $\bbone \bbe_i^T \bbx[n] = \bbone \bbx_i[n]$.
 Let $\bbepsilon[n]:=[ \bbepsilon_1[n], \dots, \bbepsilon_N [n]]^T\in \reals^N$ denote the noise values of agents on the $n$th state of the world, then we can write each $N\times N$  block of the matrix obtained in \eqref{state_variance_initial_definition_vector_2} as follows
\begin{align}
\bbE_{i,0} & \left[(\bbx[k] - \bbone \bbx_i[k]) ( \bbx[l] - \bbone \bbx_i[l])^T \right] 
%&= \bbE_{i,0}\left[  
%\left(\begin{array}{c}
%\bbx_1[k] -  \bbx_i[k]\\
%\vdots \\
%\bbx_N[k]-\bbx_i[k]
%\end{array}\right)
%\left(\begin{array}{c}
%\bbx_1[l] -  \bbx_i[l]\\
%\vdots \\
%\bbx_N[l]-\bbx_i[l]
%\end{array}\right)^T \right] \\
=  \bbE_{i,0}\left[  
\left(\bbepsilon[k] - \bbone \bbepsilon_i[k]\right) \left(\bbepsilon[l] - \bbone \bbepsilon_i[l] \right)^T
%
%\left(\begin{array}{c}
%\bbepsilon_1[k] -  \bbepsilon_i[k]\\
%\vdots \\
%\bbepsilon_N[k]-\bbepsilon_i[k]
%\end{array}\right)
%\left(\begin{array}{c}
%\bbepsilon_1[l] -  \bbepsilon_i[l]\\
%\vdots \\
%\bbepsilon_N[l]-\bbepsilon_i[l]
%\end{array}\right)^T 
\right]. \label{state_variance_initial_definition_vector_3}
\end{align}
Since initial private signals of agent $i$ are assumed to be independent of each other,
that is, $\bbE_{i,0}[\bbepsilon_i[k] \bbepsilon_i[l]] = 0$ for all $k =1, \dots, m$ and $l \neq k$, 
\eqref{state_variance_initial_definition_vector_3} is zero when $k \neq l$. When $k = l$, \eqref{state_variance_initial_definition_vector_3} is equivalent to \eqref{initial_covariance_estimates_3}. As a result, for the $N \times N$ blocks at the diagonals of $M^i_{\bbx \bbx}(0)$, we obtain \eqref{variance_vector} which is similar to its scalar counterpart given in \eqref{initial_variances}. Consider the variance of $\bbtheta$ at time $t=0$. Using \eqref{state_initial_estimate}, we obtain that $M_{\bbtheta \bbtheta}^i(0)$ is as given in \eqref{state_variance}. The initial cross covariance can also be calculated using initial mean estimates in \eqref{state_initial_estimate} and \eqref{signal_initial_estimate} in a similar way. 

Given the normal prior $\bbP_{i,0}([\bbtheta^T, \bbx^T])$ with mean estimates given by \eqref{state_initial_estimate}-\eqref{signal_initial_estimate}, the inductive hypothesis in Lemma \ref{linear_equilibrium_at_time_t_vector} is satisfied at time $t=0$. Further, by our assumption there exists a linear equilibrium action with weights $U_{i,0}$ that can be calculated by solving the set of equations in \eqref{linear_equilibrium_vector_thm}. Lemma \ref{rational_belief_updates_theorem_L_vector} already provides a way to propagate beliefs when agents play according to linear equilibrium strategy.  Furthermore, by Lemma \ref{rational_belief_updates_theorem_L_vector}, if the inductive hypothesis is true at time $t$ then it is also true at time $t+1$.
%At time $t=0$, the initial mean estimates of agent $i$ in \eqref{initial_mean_estimates} are linear combinations of private signals as given by \eqref{state_initial_estimate} and \eqref{signal_initial_estimate}. Hence, the assumption for Theorem \ref{rational_belief_updates_theorem_L_vector} is satisfied. Consequently, the mean estimates at time $t=1$ are linear combinations of the private signals with weight vectors $L_{i,1}$ and $Q_{i,1}$ following from updates \eqref{weights_recursion_x_vector} and \eqref{weights_recursion_theta_vector}, respectively. This completes the initial step of the induction. 
\end{appendices}

\bibliographystyle{unsrt}
\bibliography{bibliography}

\end{document}